\documentclass[12pt,a4paper]{article}
\usepackage[utf8x]{inputenc}
\usepackage[T1]{fontenc}

\usepackage{setspace}

\usepackage{calc} 
\usepackage{enumitem} 

\RequirePackage[authoryear]{natbib}
\RequirePackage[colorlinks,citecolor=blue,urlcolor=blue]{hyperref}
\RequirePackage{graphicx}
\usepackage{bbm, dsfont}
\usepackage{soul}
\usepackage[mathscr]{euscript}
\usepackage{mathrsfs}
\usepackage{booktabs}

\usepackage{amsmath}
\usepackage{amsfonts}
\usepackage{graphicx}
\usepackage{enumerate}

\usepackage{url} 
\usepackage{setspace}
\usepackage{multicol}
\usepackage{multirow}
\usepackage{graphicx}
\usepackage{ulem}
\usepackage{amsthm}
\usepackage{booktabs}
\usepackage{rotating} 
\usepackage{mathrsfs}
\usepackage{bbold}

\newcommand{\blind}{1}

\newcommand{\rmd}{{\rm d}}
\newcommand{\bs}{\boldsymbol}
\newcommand{\bm}{\boldsymbol}
\newcommand{\Sc}{{\fontfamily{pzc}\selectfont c}}

\newcommand{\Cov}{\mathrm{Cov}}

\newcommand{\Cor}{\mathrm{Cor}}
\newcommand{\cor}{\mathrm{Cor}}

\newtheorem{proposition}{Proposition}[section]

\theoremstyle{empty}
\newtheorem*{duplicate}{Proposition 4.1}

\begin{document}

\def\spacingset#1{\renewcommand{\baselinestretch}%
{#1}\small\normalsize} \spacingset{1}


\if1\blind
{
  \title{\bf Spectral Extremal Connectivity of Two-State Seizure Brain Waves}
  \author{Mara Sherlin D. Talento$^{a,b}$,
    Jordan Richards$^{c}$,
    Marco Pinto-Orellana$^{d}$,\\
    Raphael Huser$^a$
    and
    Hernando C. Ombao$^{a,b}$\\
    $^a$Statistics Program, Computer, Electrical and Mathematical Science \& \\ Engineering Division, King Abdullah University of Science and \\
    Technology, KSA, 23955 - 6900 \\
    $^b$Neuro-AI Laboratory, Computer, Electrical and Mathematical Science \& \\ Engineering Division, King Abdullah University of Science and \\ Technology, KSA, 23955 - 6900 \\
    $^{c}$School of Mathematics and Maxwell Institute for Mathematical Sciences, \\
    University of Edinburgh, Edinburgh, UK, EH9 3FD \\
    $^{d}$Department of Biomedical Engineering, University of California-Irvine, \\ 
    USA, 92697
    }
  \maketitle
} \fi

\if0\blind
{
  \bigskip
  \bigskip
  \bigskip
  \begin{center}
    {\LARGE\bf Spectral Extreme Connectivity Analysis of Two-State Seizure Brain Waves Data}
\end{center}
  \medskip
} \fi

\bigskip
\begin{abstract}
Coherence analysis plays a vital role in the study of functional brain connectivity. However, coherence captures only linear spectral associations, and thus can produce misleading findings when ignoring variations of connectivity in the tails of the distribution. This limitation becomes important when investigating extreme neural events that are characterized by large signal amplitudes. The focus of this paper is to examine connectivity in the tails of the distribution, as this reveals salient information that may be overlooked by standard methods. We develop a novel notion of spectral tail association of periodograms to study connectivity in the network of electroencephalogram (EEG) signals of seizure-prone neonates. We further develop a novel non-stationary extremal dependence model for multivariate time series that captures differences in extremal dependence during different brain phases, namely burst-suppression and non-burst-suppression. One advantage of our proposed approach is its ability to identify tail connectivity at key frequency bands that could be associated with outbursts of energy which may lead to seizures. We discuss these novel scientific findings alongside a comparison of the extremal behavior of brain signals for epileptic and non-epileptic patients.
\end{abstract}

\noindent%
\vfill

\newpage
\spacingset{1.9} 

\section{Introduction}\label{chap:introduction}

Seizures in neonates are known to disrupt brain development. 
Compared to adults, study of seizures in infants is more challenging
because symptoms in neonates are often more subtle \citep{Stevenson2019}. To study and assess brain function, various brain imaging modalities, such as electroencephalograms (EEGs), are used to diagnose mental and neurological disorders, e.g., seizures and epilepsy \citep[]{Muthuswamy1999DeltaTheta, Stevenson2019}. EEGs indirectly measure cortical electrical activity through non-invasive sensors placed on the scalp. In this paper, we analyze EEG recordings 
for a subset of 79 neonates in a neonatal intensive care unit (NICU) of the BAby Brain Activity Center (BABA Center), Finland \citep{Stevenson2019}. These neonates were initially recorded as having had a seizure; however, further analysis by experts confirmed seizures in only 39 patients. 
Figure~\ref{framework} shows a snippet of EEGs from the temporal lobe channels of one of the neonates. Observe that the data exhibit non-stationarity due to sudden outbursts.
This specific pattern, often seen in neonates, is called ``burst-suppression" \citep{Douglass2002BSinNeo}.

Although the presence of ``burst suppression'' patterns (see Figure~\ref{framework}) is not always associated with abnormalities in the brain \citep{Douglass2002BSinNeo}, one characteristic that can delineate hostile patterns is  connectivity (dependence between channels) in a brain network. For instance, a recent study of a surgically-removed portion of the temporal region of an epileptic patient revealed a very rare connectivity of axons at the seizure site \citep{Shapson2024Petavoxel}. Therefore, developing a statistical model to characterize brain connectivity during burst-suppression and non-burst-suppression patterns (hereafter referred to as the ``burst phase" and ``non-burst phase", respectively) can provide insights into brain functionality. Thus, in this paper, we develop a novel non-stationary multivariate model for EEG signals, that accounts for differences between states (e.g., brain phases, such as burst-suppression vs. non-burst-suppression). 


\vspace{0.25in}

\begin{figure}
    \centering
    \includegraphics[width=1.001\textwidth]{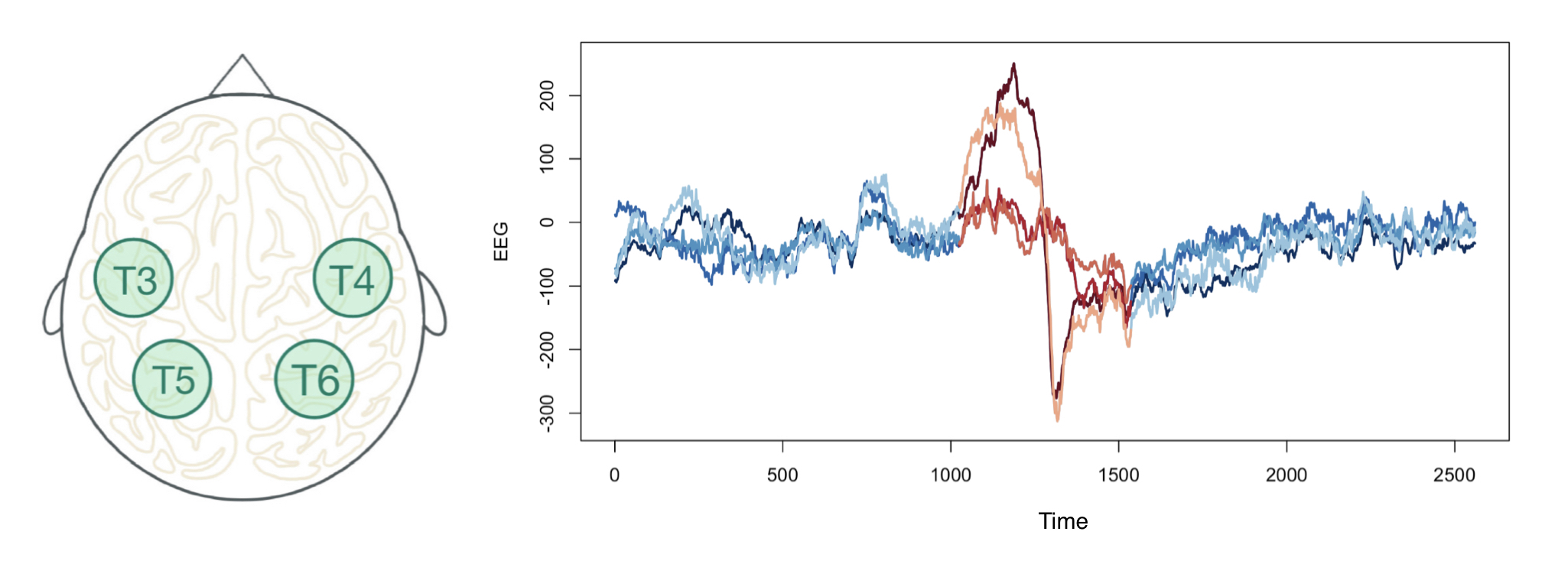}
    \caption{(Left) Channels T3-T6 in the temporal lobe, in green. (Right) Multivariate EEG Signals from channels in the temporal lobe of a neonate with normal ultrasound admitted to the NICU of the BAby Brain Activity Center, Finland. The red lines are signals during a time interval that was flagged as a burst-suppression pattern by our algorithm, described in the Supplementary Material, Section~\ref{appn2}.}
    \label{framework}
\end{figure}

\noindent {\bf EEG data.}
EEGs are often considered as a ``sum" of oscillating waveforms with random amplitudes. Thus, dependence between EEG signals can be characterized as synchronization between these oscillations.
In \cite{Srinivasan2007EEG}, standard analyses decompose EEGs into oscillations at the following five distinct frequency bands, denoted by $\Omega_{\ell}$ for ${\ell = 1,\dots, 5}$: 
(1) Delta, {$\Omega_1 := (0,4]$ Hertz} or Hz hereafter, which is dominant during sleep; 
(2) Theta, {$\Omega_2 := (4,8]$ Hz}, which is associated with drowsiness; 
(3) Alpha, {$\Omega_3 := (8,12] \text{ Hz}$}, which is associated with resting-state; (4) Beta, {$\Omega_4 := (12,30] \text{ Hz}$} and (5) Gamma, ${\Omega_5 := (30,50] \text{ Hz}}$, which are dominant when performing mental and cognitive tasks \citep{kumar2012analysis}. 
To define and estimate between-channel connectivity, the first step is to decompose the original signal into these bands. A general framework for modeling dependence between oscillations is developed in \cite{Ombao2022} which we summarize here.
Denote the neonate multi-channel EEGs  to be ${\{\bs{X}(t)\}_{t=1}^T := (X_1(t),\ldots,X_p(t))^\top}$, where 
$p$ is the number of channels and $T$ is the observation period. While Figure~\ref{framework} suggests that the EEGs are globally non-stationary, they can be modeled as weakly stationary within local (short) time blocks \citep[e.g., piecewise stationary or locally stationary;][]{Adak1998LPS, Dahlhaus2000likelihood}. Suppose that the EEGs are segmented into $B$ blocks (each of length $L$) where the EEGs are approximately weakly stationary within each block. Denote the set of time points in block $b$ to be $\mathcal{T}_b := \{(b-1)L+1 , \dots , bL\}$. Here, denote $T = LB$ to be the total number of time points across all blocks. Each signal, $X_{j}(t)$, $j = 1, \dots, p$, can be separately decomposed within each time block, $\mathcal{T}_b$, $b = 1, \dots, B$, into the standard oscillations via bandpass linear filtering (e.g., using the Butterworth filter), so that $X_j(t)  :=  \sum_{\ell = 1}^5 X_{j, \Omega_\ell}(t), t \in \mathcal{T}_b$, where $X_{j, \Omega_\ell}(t)$ is the filtered signal for channel $j$ corresponding to $\Omega_\ell$ \citep[see][]{Ombao2022}.

We develop a statistical modeling framework and a related cross-channel dependence measure which is motivated by the National Institute of Health \citep{NIH23} definition of seizure-related brain behavior; this defines the onset of a seizure as a concurrent outburst of neuronal electrical discharges across multiple neurons. 
Instead of modeling the filtered signals (or filtered time series), we study the spectral power (or squared amplitudes), which is a more direct representation of neuronal energy (or signal variance). The data analogue of the spectrum (or spectral power) are the periodograms of each EEG channel \citep{brillinger2001time}. 
As will be formalized in Section~\ref{subsec:TimeFreq}, where periodograms are directly calculated from the discrete Fourier transform (DFT), the periodogram of $\{X_j(t)\}_{t \in \mathcal{T}_b}$ at the $\Omega_\ell$-frequency band, denoted as $I^{(\ell)}_{j}(b)$, may be computed for the $b$-th block as
\begin{equation}
I^{(\ell)}_{j}(b) = \frac{1}{L}\sum_{t \in \mathcal{T}_b} \left[ X_{j,\Omega_\ell}(t) \right]^2, \;\;\; j = 1, \dots, p, \; \; b = 1, \dots, B. \label{Eq:IntroPeriodograms1}
\end{equation}
Above, it is assumed that $X_{j,\Omega_\ell}(t)$ is oscillating around its mean, 0. In this paper, we study dependence between the oscillation amplitudes of signals $\{\bs{X}(t)\}_{t \in \mathcal{T}_b}$, for all ${b = 1, \dots, B}$, through the periodograms $\bs{I}^{(\ell)}(b) = (I^{(\ell)}_{1}(b), \dots, I^{(\ell)}_{p}(b))^\top$, $\ell = 1, \dots, 5$, derived from Equation~\eqref{Eq:IntroPeriodograms1}. \cite{Guhr2021MultAmplitudes} propose a multivariate model for non-stationary time series by segmenting the long multivariate time series into distinct blocks and defining a measure for each block (similar to the approach described above). While their study did not introduce a tail dependence measure, it highlighted that non-stationarity, arising from fluctuating correlation matrices across blocks, can lead to a heavy-tailed distribution for these defined measures.

\noindent {\bf Modeling extremal dependence in spectral power.} 
To demonstrate the importance of studying tail dependence, we explore correlation between components of the periodogram (or squared-amplitudes).
While our analysis focuses on multi-channel data, for simplicity of this current discussion, we consider only two signals, say channel T4 (as $\{X_j(t)\}_{t=1}^T$) and channel T6 (as $\{X_k(t)\}_{t=1}^T$). We show here an exploratory analysis to study the strength of dependence
between the squared-amplitudes of the $\Omega_5-$band (Gamma-band) oscillations in these two channels, i.e., $\{I^{(5)}_{j}(b)\}_{b = 1}^B$ and $\{I^{(5)}_{k}(b)\}_{b = 1}^B$. In particular, we focus on dependence in the tails, when $\{I^{(5)}_{j}(b)\}_{b = 1}^B$ is ``large"; that is, we investigate the joint behavior of $(I^{(5)}_{j}(b), I^{(5)}_{k}(b)) \ | (I^{(5)}_{j}(b) > v)$ for some large value $v$. 
Let $U_{j}$ be a transformation of $\{I^{(5)}_{j}(b)\}_{b = 1}^B$ to uniform margins and let $\mathbb{F}$ be the standard Laplace distribution function. Transforming the data to a common margin avoids observing spurious dependence caused by differences in magnitude, and can be achieved via a simple rank-transformation. Here, we used the Laplace distribution, as its tails are known to converge faster to a non-degenerate distribution as we increase $v$ \citep{Rohrbeck2024Laplace}.
We estimate correlation between $\{I^{(5)}_{j}(b)\}_{b = 1}^B$ and $\{I^{(5)}_{k}(b)\}_{b = 1}^B$ (denoted by ${\rho := \Cor(\mathbb{F}^{-1}(U_{j}), \mathbb{F}^{-1}(U_{k}) )  \in [-1,1]}$), and then decompose it for the bulk (using the lower 90\% of the data; denoted by $\rho_A$) and for the tails (using the upper 10\% of the data; denoted by $\rho_B$), with 
\begin{eqnarray*}
    \rho_A & := & \Cor(\mathbb{F}^{-1}(U_{j}), \mathbb{F}^{-1}(U_{k}) \ | ( U_{j} \leq 0.9 )  ) \in [-1,1], \\
\rho_B & := & \Cor(\mathbb{F}^{-1}(U_{j}), \mathbb{F}^{-1}(U_{k}) \ | ( U_{j} > 0.9 ) )   \in [-1,1],
\end{eqnarray*}
Figure~\ref{CondAssoc} shows the estimates of $\rho_A$ and $\rho_B$ for both burst and non-burst periodograms;  we observe that connectivity tends to be stronger in the upper tail of the distribution when the signal is in the burst-phase $(\hat{\rho}_A = 0.43, \hat{\rho}_B = 0.83)$, but the converse holds for the non-burst phase $(\hat{\rho}_A = 0.26, \hat{\rho}_B = 0.16)$. This observation motivates us to develop a statistical model, based on the \underline{extremes} of periodograms, that will enable formal statistical inference on this discrepancy.

\begin{figure}
    \centering
    \includegraphics[width=1.001\textwidth]{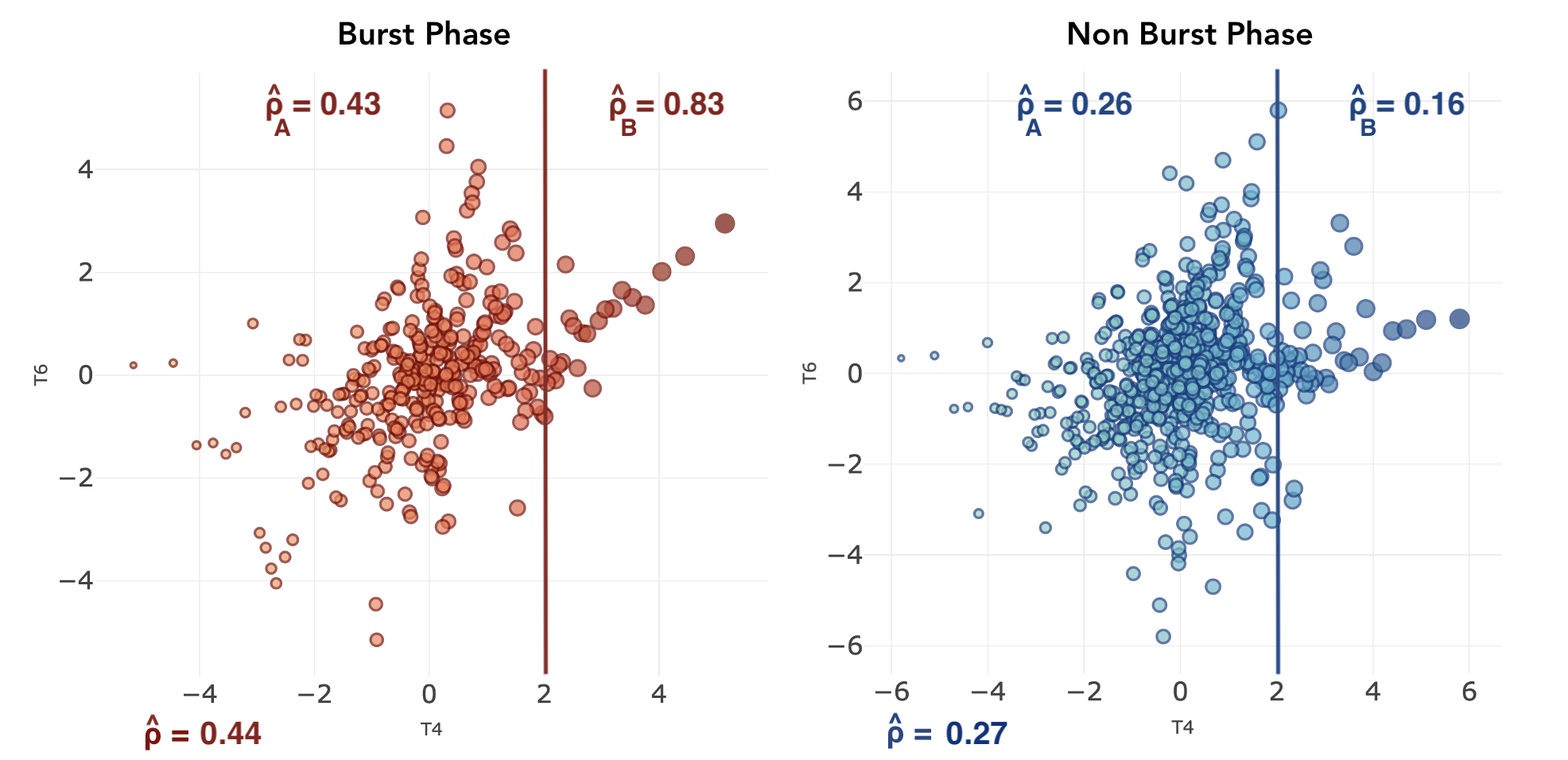}
    \caption{Periodogram of each block of pre-seizure EEG signals (on Laplace margins), classified as either burst (left) or non-burst phase (right). The coefficient provides the estimated strength of association in the entire ($\hat{\rho}$), bulk ($\hat{\rho}_A$) and tail ($\hat{\rho}_B$) of the distribution, for each phase.}
    \label{CondAssoc}
\end{figure}

Few statistical models have been developed at the interface between extreme value theory and spectral analysis of time series. Notable examples include modeling the tails of auto-periodograms using the generalized extreme value distribution, applied to EEG signals \citep{Quintero2020} and satellite data \citep{Baluev2008ExtPerio, Suveges2014ExtPerio}. While these methods are appropriate for univariate time series, they are not designed to capture dependence (across components) in multivariate time series.
More recently, \cite{Guerrero2023} provide an extreme value model for signal oscillations, introducing a conditional tail dependence measure for pairwise filtered series of seizure and pre-seizure brain signals, based on the conditional modeling framework proposed in \cite{Heffernan2004}. This was further generalized in \cite{Redondo2024}, which allowed the tail dependence parameters in \cite{Guerrero2023} to be a smooth function of time. In our approach, we develop the \textit{Spectral Extremal Connectivity} (SpExCon) model which has two key innovations: (i) we jointly model the tail behavior of multiple periodograms via a \textit{multivariate} \cite{Heffernan2004} model, rather than applying the model to pairs of filtered signals; and (ii) the proposed novel SpExCon model captures differences between two brain phases (i.e., burst vs non-burst). 
Thus, the SpExCon model is a spectral-based multivariate conditional extremes model designed to achieve two objectives: (i) modeling conditional tail dependence during two different brain phases, simultaneously; and (ii) identifying significant differences between tail dependence of burst and non-burst signals.  

The remainder of this paper is organized as follows. Section~\ref{sec:SpectralAssoc} provides an exploratory analysis 
of the observed signals using coherence, and discusses its limitation when considering changes in tail connectivity. 
Section~\ref{chap:methodology} details the SpExCon methodology, and includes the mathematical framework of the proposed tail dependence model. Section~\ref{chap:simulation} includes a simulation study that demonstrates the efficacy of our model. Section~\ref{chap:analysis} presents the key findings of the application of our model to five neonates in a NICU in Finland. The last section of our paper (Section~\ref{chap:conclusion}) summarizes the findings and highlights novel results. 

\section{Exploratory Coherence Analysis} \label{sec:SpectralAssoc}

Our prior experience with many EEG signals is that they tend to be non-stationary over very long stretches, i.e., their statistical properties, such as the mean, spectral matrix, and possibly dependence structure, can evolve over the course of the entire recording time \citep{Ombao2005slex, Fiecas2016B}. 
However, these non-stationary EEG signals can be adequately modeled as piecewise stationary processes \citep{Davis2006AutoParm}.
Consider the $b$-th time block whose time points are indexed by the set, $\mathcal{T}_b$ for $b = 1,\dots,B$, where the multi-channel EEG signal, $\{\bs{X}(t)\}_{t \in \mathcal{T}_b}$, is assumed to be weakly stationary. Hence, it admits a Cram\'er representation,
\begin{equation}
    \bs{X}(t) = \int_{-0.5}^{0.5} \bs{\mathcal A}_b(\omega) \exp(i2\pi\omega t) \rmd \bs{Z}_b(\omega), \; \; t \in \mathcal{T}_b, \label{cramer2}
\end{equation}
where $\omega \in (-0.5, 0.5)$ is the standardized frequency interval, $\bs{\mathcal A}_b(\omega)$ is the $p \times p$ complex-valued transfer function matrix on time block $b$, and $\rmd \bs{Z}_b(\omega) := (\rmd Z_{1,b}(\omega), \dots, \rmd Z_{p,b}(\omega))^\top$ is a $p$-dimensional random process with zero mean and whose covariance structure satisfies $ \Cov (\rmd \bs{Z}_b(\omega), \rmd \bs{Z}_b(\lambda)) = 0$ when $\omega \ne \lambda$ and  $ \Cov (\rmd \bs{Z}_b(\omega), \rmd \bs{Z}_b(\omega)) = {\mathbb I} \, \rmd\omega$, where ${\mathbb I}$ is the identity matrix. Here, the $p \times p$ Hermitian spectral matrix of $\{\bs{X}(t)\}_{t \in \mathcal{T}_b}$ is ${{\bm f}_b(\omega) := \bs{\mathcal A}_b(\omega) \bs{\mathcal A}^*_b(\omega)}$, where $\bs{\mathcal A}^*_b(\omega)$ denotes the complex-conjugate transpose of $\bs{\mathcal A}_b(\omega)$. The standard dependence measure between the $j$-th and $k$-th EEG channels in Equation~\eqref{cramer2}, i.e., coherence, is defined through the random coefficients $\rmd\mathcal{Z}_{j,b}(\omega) := \bs{\mathcal{A}}_{j,b}(\omega)\rmd \bs{Z}_{b}(\omega)$ for $j = 1, \dots, p$, where $\bs{\mathcal{A}}_{j,b}(\omega)$ is the $j$-th row of the $\bs{\mathcal A}_b(\omega)$ matrix.
An intuitive interpretation in terms of oscillations is explained in \cite{Ombao2022}, where the $\omega$-oscillation of $\{X_j(t)\}_{t \in \mathcal{T}_b}$ is $X_{j,\omega}(t) = \exp(i2\pi\omega t)\rmd \mathcal{Z}_{j,b}(\omega)$, for all $t \in \mathcal{T}_b$. 

Coherence between time series $\{X_j(t)\}_{t \in \mathcal{T}_b}$ and $\{X_k(t)\}_{t \in \mathcal{T}_b}$ in block $b$ and frequency band $\Omega_\ell$, is then defined as
\begin{eqnarray}
    \varrho_{jk}^{(\ell)}(b) & := & \frac{1}{|\Omega_\ell|} \int_{\{\omega : s\omega \in \Omega_\ell\}} \varrho_{jk}(b,\omega) \rmd\omega, \text{ where} \\
    \varrho_{jk}(b,\omega) & := & | \cor(X_{j, \omega}(t), X_{k, \omega}(t); t \in \mathcal{T}_b)|^2 = |\cor(\rmd \mathcal{Z}_{j,b}(\omega), \rmd \mathcal{Z}_{k,b}(\omega))|^2, \notag 
    \label{coherence2}
\end{eqnarray}
and $s$ is the sampling rate of the signal. Coherence (as well as partial coherence) has been widely used to characterize dependence in a brain network \citep{Fiecas2011A, Park2014estimating}.
One limitation of coherence is that it may miss features of dependence in the tails of the distribution. Figure~\ref{CondAssoc} demonstrated the important point that dependence strength between channels can vary between the tail and the bulk of the distribution, and can differ between burst and non-burst phases. To further investigate the limitations of coherence, we conduct an exploratory data analysis for one neonate. The EEG signals are segmented into 1000 quasi-stationary blocks, with each block of length two seconds and sampling rate $s = 256$ Hz, or 256 observations per second; hence $L = 512$. For each block, we estimate coherence at the Delta band, $\varrho_{jk}^{(1)}(b)$. Each block is labeled either as burst or non-burst (see the Supplementary Material, Section~\ref{appn2}, for details on the labeling algorithm). 
Our analysis focuses on the temporal lobes (channels T3, T4, T5, and T6) which is the common focal point of 95\% of seizures \citep{Medvedev2011Gamma}. 
\begin{figure}
    \centering
    \includegraphics[width=0.99\textwidth]{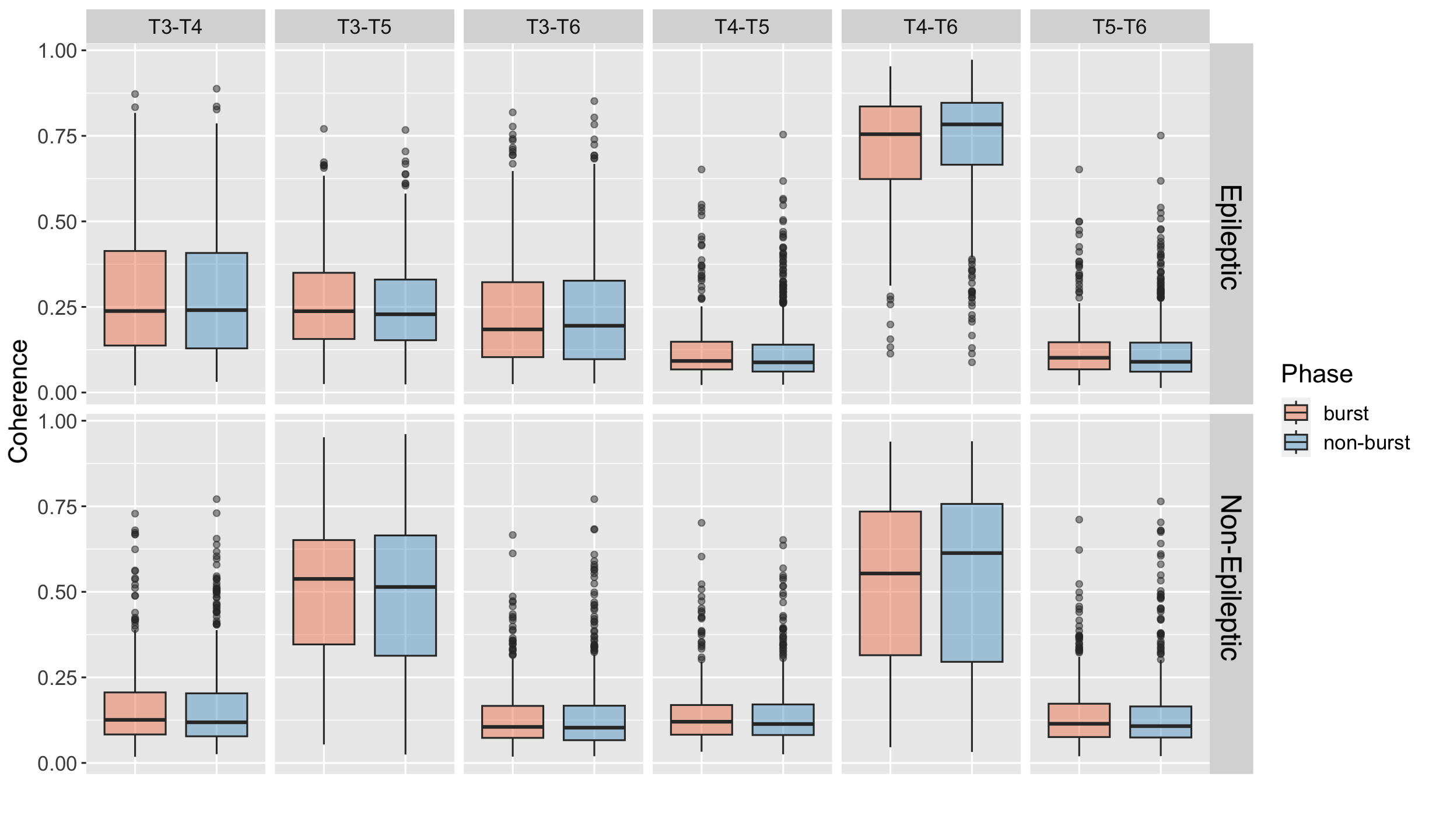}
    \caption{Boxplots of pairwise coherence estimates for channels T3--T6 of an epileptic (top) and a non-epileptic (bottom) patient. Coherences are estimated for burst (red) and non-burst (blue) signals of the Delta band, (0--4) Hz.}
    \label{DistCoh}
\end{figure}
Boxplots of the pairwise coherence estimates for the Delta-band ($\Omega_1$), i.e., $\hat{\varrho}_{jk}^{(1)}(b)$ for $b = 1, \dots, 1000$, displayed in Figure~\ref{DistCoh}, 
show no visually pronounced differences in the
distribution of the estimated coherence between the two phases. This analysis of coherence could lead one to conclude that dependence between the Delta oscillations at the temporal channels does not differ between the burst and non-burst phases. However, as observed in Figure~\ref{CondAssoc} and Section~\ref{chap:analysis}, differences in the strength of association can be detected when using appropriate measures tailored to the tails. This highlights the limitations of current methods and serves as our {motivation} to develop a new approach, that we call SpExCon (with details in Section 3), for analyzing spectral connectivity in brain signals which obscure the significant differences that are revealed only in the tails. 

\section{The SpExCon Method} \label{chap:methodology}

We describe, in Section~\ref{subsec:TimeFreq}, the time-frequency analysis under the \cite{Dahlhaus2000likelihood} model. 
Then, in Section~\ref{subsec:Model}, we develop our model for non-stationary multivariate spectral extremal connectivity (SpExCon). Furthermore, we test the significance of parameter estimates through a bootstrap resampling procedure discussed in Section~\ref{subsec:bootstrap}.  A flowchart schematic of the entire proposed SpExCon method is given in Figure~\ref{flowcart}. 


\begin{sidewaysfigure}
    \centering
    \includegraphics[width=0.99\textwidth]{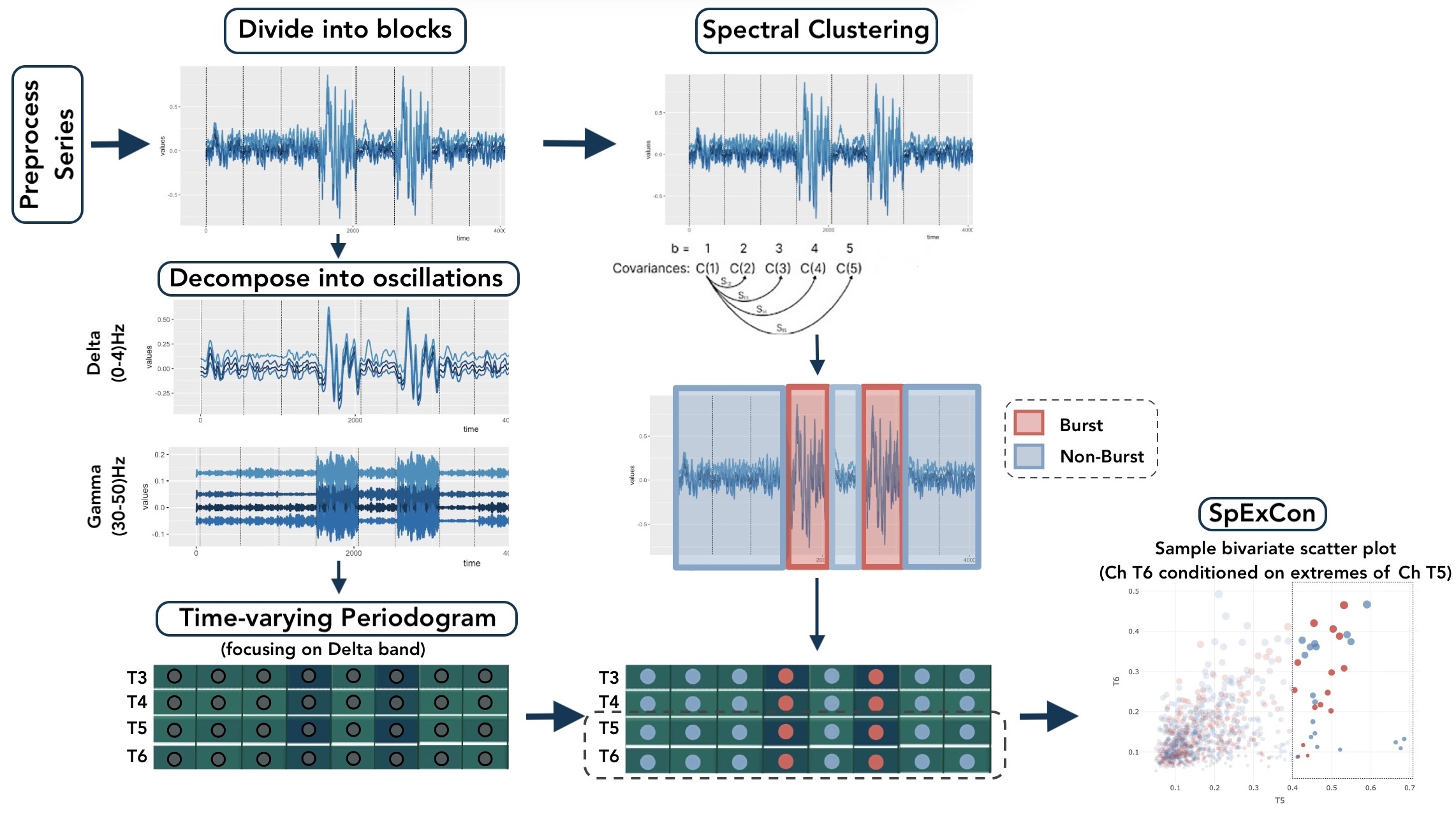}
    \caption{Schematic of the SpExCon modeling framework.}
    \label{flowcart}
\end{sidewaysfigure}

\subsection{Time-Frequency Analysis for Non-stationary Series} \label{subsec:TimeFreq}

The Cram\'er representation in Equation~\eqref{cramer2} was generalized by \cite{Priestley1967power} to allow the transfer function to change with time---we denote it here by $\bs{\mathcal A}(t,\omega)$. This was later refined in \cite{Dahlhaus2000likelihood}, showing that $\bs{X}(t)$ can be expressed as
\begin{equation}
    \bs{X}(t) = \int_{-0.5}^{0.5} \bs{\mathcal A}(u, \omega) \exp(i2\pi\omega t) \rmd \bs{Z}(\omega), \; \; \; t \in 1, \dots, T, \label{dahlhaus}
\end{equation}
where the transfer function $\bs{\mathcal A}(u, \omega)$ is a function of both frequency $\omega \in (-0.5, 0.5)$ and rescaled time $u = t/T \in [0,1]$. This framework allowed for the construction of mean-squared consistent estimators for the time-varying spectrum ${\bm f}(u, \omega) = \bs{\mathcal A}(u, \omega) \bs{\mathcal A}^*(u, \omega)$. 

A similar class of models treats the nonstationary multichannel EEG $\{{\bm X}(t)\}_{t = 1}^T$ as piecewise stationary. This class of models include the piecewise-stationary autoregressive models, the smooth localized complex exponentials (SLEX) model of nonstationary processes \citep{Ombao2005slex}, and the piecewise Cram\'er process, developed by \cite{Adak1998LPS}, which we adopt here. Under this model, the time series is segmented into piecewise stationary blocks. 
For simplicity, we assume that all blocks $b$ ($b=1, \ldots, B$) have equal length $L$. Block $b$ is given the label $D_b = 1$ if it is a ``burst" phase and $D_b = 0$ if it is a ``non-burst" phase. Labeling of the phases is performed using spectral clustering, see \cite{Narula2021}, and is detailed in the Supplementary Material, Section~\ref{appn2}. We collect all ``burst" blocks in $\mathcal{B}_1 := \{b : D_b =1 \}$, and similarly for the non-burst blocks, $\mathcal{B}_0$. Furthermore, we assume that the burst and non-burst phases have different spectral properties (and hence different Cram\'er representations). However, we assume that all blocks labeled as realizations of the ``burst" phase are independent copies of a particular process; and the same holds for blocks labeled as ``non-burst". 

Let the transfer function for the burst and non-burst phases be, respectively,  
$\bs{\mathcal A}^{1}(\omega)$ and $\bs{\mathcal A}^{0}(\omega)$ and the spectral matrix for the burst and non-burst phases be, respectively, 
\begin{equation*}
{\bm f}^1(\omega)  :=  \bs{\mathcal A}^{1}(\omega)\bs{\mathcal A}^{1*}(\omega) \ \ {\mbox{and}} \ \ 
{\bm f}^0(\omega)  :=  \bs{\mathcal A}^{0}(\omega)\bs{\mathcal A}^{0*}(\omega).
\end{equation*}
That is, we let the function $\bs{\mathcal A}(u, \omega)$ in \eqref{dahlhaus} be $\bs{\mathcal A}^{1}(\omega)$ for all $u = t/T$ such that $t$ is in the burst phase, i.e., $t \in \{\mathcal{T}_b :  D_b = 1, b = 1,\dots, B\}$ (and similarly for $\bs{\mathcal A}^{0}(\omega)$ and non-burst phases). 
From Equation~\eqref{dahlhaus}, we define the transfer function matrix in block $b$ as $\bs{\mathcal A}(t/T,\omega) := \bs{\mathcal A}_b(\omega) $ for all $t\in\mathcal{T}_b$, such that, in \eqref{cramer2},
\begin{equation}
    \bs{\mathcal A}_b(\omega) = \bs{\mathcal A}^{1}(\omega) D_b + \bs{\mathcal A}^{0}(\omega) (1 - D_b). \label{Abw}
\end{equation}
Hence, the true spectrum at block $b$ is ${\bm f}_b(\omega) := {\bm f}^1(\omega) D_b + {\bm f}^0(\omega) (1-D_b)$.

We now calculate the periodogram matrix for block $b$ \citep[i.e., the data analogue of the spectrum; see][]{brillinger2001time} for all time $t \in \mathcal{T}_b = \{(b-1)L + 1, \ldots, bL\}$.
The periodogram matrix for block $b$ at a fundamental frequency $\omega_r \in - \lfloor (L/2 -1) \rfloor, \ldots, -1, 0, 1, \ldots, \lfloor L/2 \rfloor$, denoted by ${\bm I}(b,\omega_r)$, consists of diagonal elements $I_{j}(b,\omega_r)$ (i.e., auto-periodograms) and off-diagonal elements $I_{jk}(b,\omega_r)$ (i.e., cross-periodograms). It can be obtained as 
\begin{equation}
   I_{j}(b,\omega_r) = a_{j,b}(\omega_r) a_{j,b}^{*}(\omega_r) \ \ {\mbox{and}} \ \ I_{jk}(b,\omega_r) = a_{j,b}(\omega_r) a_{k,b}^{*}(\omega_r), \label{TVPerio}
\end{equation}
where $a_{j,b}(\omega_r)$ denotes the discrete Fourier transform of $\{X_j(t)\}_{t \in \mathcal{T}_b}$, defined as
\[
a_{j,b}(\omega_r) := \frac{1}{\sqrt{L}} \sum_{t \in \mathcal{T}_b} X_j(t) \exp(-i 2 \pi \omega_r (t - [(b-1)L + 1])).
\]

\noindent To get consistent estimates, periodograms are typically smoothed. In this study, smoothing is performed by averaging the periodograms of neighboring frequencies. Let $n$ be such that $2n + 1 \ll L$. In this study, we take $n = 5$ for a block length of $L=512$. The smoothed (or averaged) auto-periodogram is then defined as $\tilde{I}_{j}(b,\omega_r) := \frac{1}{2n+1} \sum_{k = -n}^n I_{j}(b,\omega_r + k/L)$ \citep[see][for details]{Shumway2000time}.
We then compute the periodogram at the standard frequency bands $\Omega_{\ell}$ (for $\ell=1, \ldots, 5$) commonly used in practice, as follows.
For a given sampling rate $s$ and block length $L$ (for these EEG data, $s=256$ and $L=512$), the fundamental Fourier frequencies that correspond to the band $\Omega_{\ell} = (\Omega_{\ell}^{(1)}, \Omega_{\ell}^{(2)}]$ Hz is the set ${Q_{\ell} = \{\omega_r:  \ \ s\omega_r \in [-\Omega_{\ell}^{(2)}, -\Omega_{\ell}^{(1)}) \ \bigcup \ (\Omega_{\ell}^{(1)}, \Omega_{\ell}^{(2)}] \}}$.
Then the smoothed periodogram for the EEG channel $\{X_j(t)\}_{t \in \mathcal{T}_b}$ at block $b$ and frequency band $\Omega_{\ell}$ is 
\begin{equation}
{I}_{j}^{(\ell)}(b) := \frac{1}{|Q_{\ell}|} \sum_{
\omega_r \in Q_{\ell}} \tilde{I}_{j}(b, \omega_r). \label{TVPerio2}
\end{equation}
For simplicity, we use the same notation ${I}_{j}^{(\ell)}(b)$ as in \eqref{Eq:IntroPeriodograms1} and assume that the auto-periodogram is smoothed. 
The periodogram ${I}_{j}^{(\ell)}(b)$ measures the variance (energy) contribution of the frequency band $\Omega_{\ell}$ in $X_j$ to the total variance of the said $j$-th channel.

\subsection{Non-stationary Multivariate Conditional Spectral Extremes} \label{subsec:Model}
\label{sec3:DEC}
To study differences in the joint extremal behavior of the energy at frequency band $\Omega_{\ell}$ during the burst, $\mathcal{B}_1$, and non-burst phases, $\mathcal{B}_0$, we study extremes of the periodogram vector ${\bs I}^{(\ell)}(b) = (I_1^{(\ell)}(b), \dots, I_p^{(\ell)}(b))^\top$ (for $b \in \mathcal{B}_1$ and for $b \in \mathcal{B}_0$, respectively). 
To this end, we construct a two-state model for multivariate extremal dependence. Although many appropriate multivariate extremal dependence models have been developed \citep[see, e.g., the discussion by][]{papastathopoulos2023statistical}, we choose to exploit the multivariate conditional extremes framework of \cite{Heffernan2004} \citep[see, also,][]{heffernan2007, keef2013estimation}, as this framework can flexibly capture both positive and negative extremal dependence, and its parameter estimates are easy to interpret (relative to other models for multivariate extremes). Moreover, \cite{Guerrero2023} have already illustrated the efficacy of the \cite{Heffernan2004} framework for brain connectivity analysis, and \cite{richards2023joint} and \cite{tendijck2023modeling} have shown that the conditional model is useful for capturing mixture structures in (environmental) data.

Following \cite{keef2013estimation}, we consider a generic $p$-dimensional vector with Laplace margins by $\bs{Y} = (Y_1, \dots, Y_p)^\top \in\mathbb{R}^p$ (in our approach, these correspond to the block-wise periodograms, $\bs{I}^{(\ell)}(b)$ transformed to the Laplace scale). Let $\bs{Y}_{-q}\in\mathbb{R}^{p-1}$ be the vector $\bs{Y}$ with its $q^{th}$ component removed. The classical conditional model assumes that there exist some normalizing parameters $\boldsymbol{\alpha}_{-q} := \{{\alpha}_{j|q} : j \in (1,\dots,p) \setminus q\} \in [-1,1]^{p-1}$ and $\boldsymbol{\beta}_{-q} := \{\beta_{j|q} : j \in (1,\dots,p) \setminus q\} \in [0,1)^{p-1}$ such that, for $\bs{z}\in\mathbb{R}^{p-1}$ and $y>0$,
\begin{equation}
    \label{Eq:H+T}
    \mathbb{P}\left\{ \frac{\bs{Y}_{-q} - \bs{\alpha}_{-q}Y_{q}}{Y_{q}^{\bs{\beta}_{-q}}} \leq \bs{z}, Y_{q} - v > y \; \Big| \; Y_{q} > v \right\} \rightarrow G_{-q}(\bs{z})\exp(-y), \text{ as } v \rightarrow \infty, 
\end{equation}
where operations are taken componentwise and $G_{-q}(\cdot)$ is some non-degenerate $(p-1)$-variate distribution function. Modeling with Equation \eqref{Eq:H+T} follows by assuming that the asymptotic limit holds in equality for some fixed but large value $v>0$; in this case, conditional on $Y_q>v$, the vector $\bs{Y}_{-q}$ can be seen to follow the heteroskedastic regression model $(\bs{\alpha}_{-q}{Y}_{q} + Y_q^{{\bs \beta}_{-q}}{\bs Z}_{-q})\mid (Y_q>v)$, with residual vector ${\bs Z}_{-q} \sim G_{-q}$.

Values of $\boldsymbol{\alpha}_{-q}$ and $\boldsymbol{\beta}_{-q}$ determine the strength and type of tail dependence exhibited by $\bs{Y}$. For example, when the $j$-th entry of $\boldsymbol{\alpha}_{-q}$ is in $(0,1)$ or $(-1,0)$ we have, respectively, positive or negative extremal association between $Y_q$ and the $j$-th component of $\bs{Y}_{-q}$, with the strength of dependence increasing with the parameters' magnitude.
A positive extremal association indicates that high values of one variable are likely to coincide with high values of another. Conversely, a negative extremal association implies that high values of one variable tend to occur alongside low values of another.
When values of ${\alpha}_{j|q}$ and ${\beta}_{j|q}$ are equal to one (minus one) and zero, respectively, then the corresponding components, ${Y}_{j|q}$ and $Y_q$, are positively (negatively) asymptotically dependent; they are asymptotically independent, otherwise. 
For further details on the concept of asymptotic (tail) dependence and independence, see \cite{Heffernan2004} or \cite{huser2022advances}.

We use model \eqref{Eq:H+T} to study extremal dependence in the band-specific periodogram vectors. That is, we replace the generic vector $\bs{Y}$ in \eqref{Eq:H+T} with $\bs{I}^{(\ell)}(b) := (I_{1}^{(\ell)}(b), \dots, I_{p}^{(\ell)}(b))^\top$, transformed to Laplace margins, at frequency band $\Omega_\ell,\ell = 1,\dots,5,$ and for block $b = 1,\dots,B$. For brevity, we drop the superscript $\ell$ from notation $\bs{I}^{(\ell)}(b)$, as we hereafter perform all dependence modeling separately at each frequency band. We denote by $\bs{I}_{-q}(b)$ the vector $\bs{I}(b)$ with its $q^{th}$ component removed, and recall that the random variable $D_b$ takes value $d$ if $b\in \mathcal{B}_d$ for $d=0,1$. To account for differences in the burst and non-burst phase extremal dependence structures of $\bs{I}(b)$, we incorporate $D_b=d$ as a covariate in a non-stationary conditional extremes model; see, e.g., \cite{winter2016modelling} and \cite{richards2023modern}. Let 
\begin{equation}
    v:=\mathbb{F}^{-1}(\kappa)>0 \label{kappa}
\end{equation}
be a large threshold, taken to be the $\kappa$-th quantile of the standard Laplace distribution, $\mathbb{F}$, for $\kappa \in (0,1)$. Then, for block $b=1,\dots,B$ and $y>v$, our phase-dependent conditional model is
\begin{align}
     \bs{I}_{-q}(b) \; \bigg| \bigg(\; I_{q}(b) = y  ; D_b = d \bigg) 
    &= \boldsymbol{\alpha}_{-q}^{(d)} y +  y^{\boldsymbol{\beta}_{-q}} 
    \bs{Z}^{(d)}_{-q}\nonumber\\
    &= \tanh\left(\bs{a}_{-q} +  \bs{\Delta}_{-q} d  \right) y +  y^{\boldsymbol{\beta}_{-q}} 
    \bs{Z}^{(d)}_{-q},
        \label{DEC}
\end{align}
with operations taken {componentwise}, parameter vectors 
\begin{eqnarray*}
    \boldsymbol{a}_{-q} & = & \{a_{j|q} : j \in \{1, \dots, p \} \setminus q \} \in \mathbb{R}^{p-1}, \\
    \boldsymbol{\Delta}_{-q} & = & \{\Delta_{j|q}: j \in \{1, \dots, p \} \setminus q \} \in \mathbb{R}^{p-1}, \\
    \boldsymbol{\beta}_{-q} & = & \{ \beta_{j|q}:  j \in \{1, \dots, p \} \setminus q \} \in [0,1)^{p-1},
\end{eqnarray*}
and $d$-dependent residual vector $\bs{Z}^{(d)}_{-q} = \{ Z_{j|q}^{(d)}:  j \in \{1, \dots, p \} \setminus q \}  \sim G_{-q}^{(d)}$.
Here, the difference in strength of extremal dependence of the periodogram vector, between the burst ($d = 1$) and non-burst phase ($d = 0$), is determined through the parameter vector $\boldsymbol{\Delta}_{-q}$, which can be found in the vector of $d$-dependent extremal dependence parameters  
\begin{equation}
    \boldsymbol{\alpha}_{-q}^{(d)} = \tanh\left(\boldsymbol{a}_{-q} +  \boldsymbol{\Delta}_{-q} d \right) = \{ \alpha_{j|q}^{(d)}:  j \in \{1, \dots, p \} \setminus q \} \in [-1,1]^{p-1}.   \label{psi}
\end{equation}

Here, we have allowed the $\boldsymbol{\alpha}_{-q}$ parameter vector and the residual distribution  $G^{(d)}_{-q}$ of the \cite{Heffernan2004} model to vary with phase, but have specified a common $\boldsymbol{\beta}_{-q}$ parameter across (burst or non-burst) phases. While a simple extension of our model will allow for a $d$-dependent $ \boldsymbol{\beta}_{-q}$, we propose this simplified model to maximise interpretability of the parameter estimates. Moreover, an unreported sensitivity analysis of this choice was conducted, and found to have negligible impact on the results.

Inference for model \eqref{DEC} proceeds via maximum likelihood estimation, under the working parametric assumption that $G_{-q}^{(d)}$ is multivariate Gaussian with $d$-dependent mean vector and positive-definite covariance matrix. We estimate the marginal distributions of $\{I_j(b)\}_{b \in \mathcal{B}_d}$ empirically for each phase-type, $d=0,1$, and these are then used to standardize the time-varying auto-periodogram to standard Laplace margins \citep[see][]{keef2013estimation}. Exceedances above $v$, the $\kappa$-th standard Laplace quantile, are then used to fit model \eqref{DEC}; see, e.g., \cite{winter2016modelling} for full inference details. \cite{richards2023modern} advocate the use of QQ-plots to assess goodness-of-fit of the multivariate extremal dependence model, and to optimise the quantile level $\kappa$. We follow their lead and use estimates of model \eqref{DEC} to simulate, for each conditioning channel $q$, an aggregate random variable
\begin{eqnarray}
    R_b^{(-q)} : = \sum_{j \in \{1, \dots, p \} \setminus q } \left\{I_j(b) \mid(I_{q}(b) > v)\right\}. \label{aggR}
\end{eqnarray}
Note that we simulate from $\bs{I}_{-q}(b)\mid(I_{q}(b) > v)$ by drawing from ${\bs{I}_{-q}(b)\mid(I_{q}(b) > v, D_b=d)}$ where the probability $\mathbb{P}(D_b=d),d=0,1$, is estimated empirically. Moreover, to simulate $\bs{Z}_{-q}^{(d)}$, we do not use the working assumption that $G^{(d)}_{-q}$ is Gaussian; we instead draw from the empirical residual distribution, with the residuals computed by substituting observations of the multivariate periodogram and estimates of parameter vectors $\bs{\alpha}_{-q}^{(d)}$ and $\bs{\beta}_{-q}$ into Equation~\eqref{DEC}.
Goodness-of-fit of the SpExCon model is then assessed visually, by noting deviation from the diagonal line in a QQ-plot of simulated against empirical values of $\{R_b^{(-q)}\}_{b=1}^B$.



\subsection{Bootstrap Procedure} \label{subsec:bootstrap}

To quantify parameter estimation uncertainty, we employ a two-stage resampling technique. As noted in Section~\ref{subsec:TimeFreq}, the EEGs are segmented into distinct blocks, $b = 1, \dots, B$. The first-stage of our bootstrap is a random selection of $B$ blocks with replacement from the pool of ``burst" ($\mathcal{B}_1$) and ``non-burst" blocks ($\mathcal{B}_0$), based on their corresponding probabilities, $\mathbb{P}(D_b = d) = \mathbb{P}(b \in \mathcal{B}_d) := \frac{|\mathcal{B}_d|}{B}$, $d = 0,1$. 
After taking a sample of $B$ blocks, denoted by $\mathcal{B}^{boot}$, we perform a stationary bootstrap \citep{Politis1994stationary} within each block in $\mathcal{B}^{boot}$. 
This involves concatenating sub-blocks from the original time series, starting at a random time point $t_m$ and lasting for a random duration $L_m$ drawn from geometric distribution with mean $h^*$, i.e., $L_m \sim \text{Geom}(1/h^*)$. The mean of this geometric distribution, $h^*$, is usually chosen as the maximum lag that has significant autocorrelation. This procedure is repeated until the bootstrap time series is of the same length as the original, i.e., $L$; see Figure~\ref{flowcart2} for an illustration.

\begin{figure}
    \centering
    \includegraphics[width=0.99\textwidth]{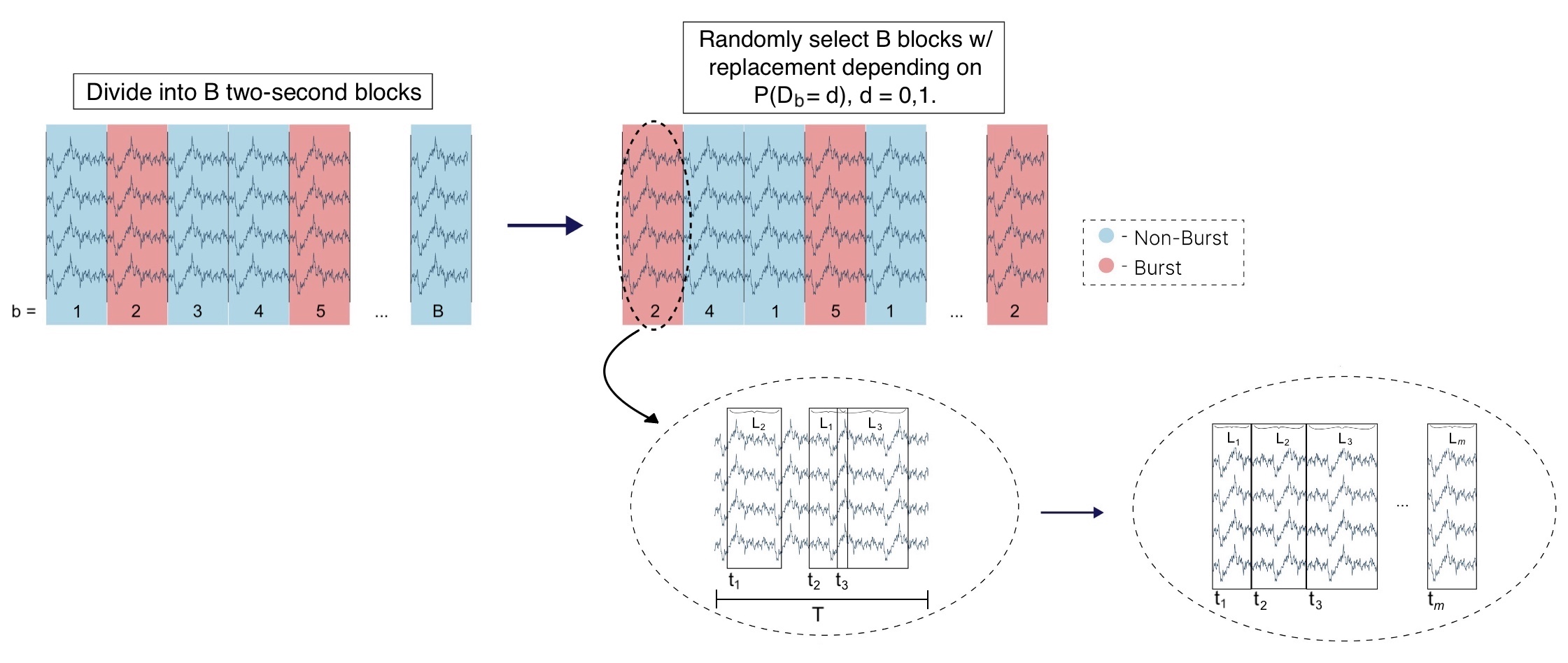}
    \caption{Schematic of the bootstrap resampling framework.}
    \label{flowcart2}
\end{figure}




We then refit our model to each bootstrapped EEG dataset.
Model estimates yield the bootstrap sampling distributions of the parameters, $\boldsymbol{\alpha}_{-q}^{(d)} \in [-1,1]^{p-1}$, ${\bs \beta}_{-q} \in [0, 1)^{p-1}$, and $\boldsymbol{\Delta}_{-q} \in \mathbb{R}^{p-1}$, which are then used to obtain percentile-based bootstrap confidence intervals. The significance level of the confidence intervals is adjusted to control for false coverage-statement rate (FCR) using the approach of \cite{Benjamini2005FDR}.


\section{Simulation Study} \label{chap:simulation}

We now perform a simulation study to show the efficacy of the SpExCon method in identifying changes in the spectral extremal dependence structure of multivariate time series. Section~\ref{subsec:SimSettings} details
a time series model for EEG data with known extremal dependence structure at the frequency oscillation bands, $\Omega_\ell$ for $\ell = 1, \dots, 5$. We use this model in Section~\ref{subsec:SimResult}, which presents the result of our simulation studies, to showcase the accuracy of estimating the parameter $\Delta$ (i.e., comparing the true value with the empirical distribution of the estimator), as well as the ability of SpExCon to identify burst and non-burst phases.

\subsection{True underlying model} \label{subsec:SimSettings}
To mimic the burst and non-burst phases of EEG data, we simulate from a multivariate process mixture which exhibits two different (phase-dependent) covariance structures; we use a mixture of second-order autoregressive processes (AR(2)), defined for $t \in 1, \dots, T$, as
\begin{equation}
    \bs{X}(t) = \bs{C}(t)\bs{O}(t) + \bs{W}(t),   \label{rawsim}
\end{equation}
where $\bs{C}(t) \in \mathbb{R}^{p \times 5}$ is constant within a block, so we write with some abuse of notation, $\{\bs{C}(t)\}_{t \in \mathcal{T}_b} := \bs{C}_b$. The vector $\bs{O}(t) = [O_1(t), \ldots, O_5(t)]^\top$ contains independent real-valued time-series $\{O_{\ell}(t)\}_{t = 1}^T$, $\ell = 1, \dots, 5$, each of which is an AR(2) process that corresponds to a frequency-band-specific oscillation, $\Omega_\ell$. Finally, $\bs{W}(t) \in \mathbb{R}^{p}$ is an independent vector of white-noise processes, such that $\mathbb{E}[W_j(t)] = 0$ and $\mathbb{V}[W_j(t)] <\infty$, for all $j = 1, \dots, p$. 
The AR(2) process $\{O_{\ell}(t)\}_{t = 1}^T$ is parameterized by $\theta_{1,\ell} \in \mathbb{R}$ and $\theta_{2} \in \mathbb{R}$, such that, for all $t = 1, \dots, T$,
$O_{\ell}(t):= \theta_{1,\ell}O_{\ell}(t-1) + \theta_{2}O_{\ell}(t-2) + \varepsilon(t),$
where $\varepsilon(t)$ is a zero-mean white noise process. 
The AR(2) parameters are defined as $\theta_{1,\ell} = 2\exp(-M)\cos(2\pi |\omega_\ell|)$ and $\theta_{2} = -\exp(-2M)$, where $\omega_{\ell} \in \Omega_{\ell}/s$, $M = 1.05$, and $s = 256$ Hz for the sampling rate \citep[see details in,][]{Ombao2022}. 
Specifically, we consider $\omega_\ell$ to be ${\{\omega_\ell : s\omega_\ell \in \{2, 6, 10, 20, 40\} \}}$.
Notice that $s\omega_\ell$ belongs to any of the frequency bands $\Omega_\ell$ defined in Section~\ref{chap:introduction}.

For this mixture process, we can show that the tail dependence of its corresponding multivariate block-wise periodogram, $\bs{I}^{\ell}(b)$, at frequency band $\Omega_{\ell}$, is reflected through the tail dependence of elements of $\bs{C}_b$. In the Supplementary Material, Section \ref{appn1}, we show that tail dependence in the periodograms of $\bs{X}(t)$ in Equation~\eqref{rawsim} is driven solely by tail dependence in the random matrix $\bs{C}_b$. By defining a suitable model for $\bs{C}_b$, we can induce a time series $\{\bs{X}(t)\}_{t = 1}^T$ with known spectral tail dependence properties; that is, we can simulate from a multivariate time series model for which we know the true SpExCon parameters. 
\begin{proposition} \label{Prop:Simulation}
   Consider a `locally stationary' $p$-dimensional time series $\{\bs{X}(t)\}_{t=1}^T$ so that the time series is approximately stationary for distinct blocks of length $L$, indexed by time points $\mathcal{T}_b := \{(b-1)L+1 , \dots , bL\}, b = 1, \dots, B$, such that $T = LB$. Specifically, let ${\bs{X}(t) = \bs{C}(t)\bs{O}(t) + \bs{W}(t),}$ for $t \in \{\mathcal{T}_b\}_{b = 1}^B$, where $\bs{C}(t) \in \mathbb{R}^{p \times 5}$ is assumed to be constant within a time block $\mathcal{T}_b$, i.e., $\bs{C}((b-1)L+1) = \dots = \bs{C}(bL)$, for $b \in \mathcal{B}_d$, and $d = 0, 1$, such that $\mathcal{B}_0 \cap \mathcal{B}_1 = \emptyset$ and $\mathcal{B}_0 \cup \mathcal{B}_1 = \{1, \dots, B\}$. Moreover, assume the vector $\bs{O}(t) = [O_1(t), \ldots, O_5(t)]^\top$ is comprised of mutually-independent AR(2) processes, and $\bs{W}(t) \in \mathbb{R}^{p}$ is an independent vector of white-noise processes.
    Define the vector $\tilde{\bs{C}}_{\ell}(b) := (\tilde{C}_{1\ell}(b), \dots, \tilde{C}_{p\ell}(b))^\top$, for a fixed $\ell \in \{1, \dots, 5\}$ and $b \in \mathcal{B}_d$, such that 
    $\tilde{C}_{j\ell}(b) := \{(C_{j\ell}(t))^2: t \in \mathcal{T}_b\} $, for $j = 1, \dots, p$. 
    Assume that $\{\tilde{\bs{C}}_{\ell}(b)\}_{b \in \mathcal{B}_d}$, for $d = 0, 1$, has a Gaussian copula with exponential margins, and let $\rho^{(\ell,d)}_{qj} > 0$ denote the correlation between $\{\tilde{C}_{q\ell}(b)\}_{b \in \mathcal{B}_d}$ and $\{\tilde{C}_{j\ell}(b)\}_{b \in \mathcal{B}_d}$. Then, for a \underline{fixed $\ell$}, the SpExCon tail dependence parameters in \eqref{DEC} are $\alpha^{(d)}_{j|q} = (\rho^{(\ell,d)}_{qj})^2$, $a_{j|q} = \tanh^{-1}\{(\rho^{(\ell,0)}_{qj})^2\}$, $\beta_{j|q} = 1/2$ and $\Delta_{j|q} = \tanh^{-1}\{(\rho^{(\ell,1)}_{qj})^2\} - \tanh^{-1}\{(\rho^{(\ell,0)}_{qj})^2\}$.
\end{proposition}
The proof of Proposition~\ref{Prop:Simulation} can be found in the Supplementary Material, Section~\ref{appn1}.


\begin{figure}
    \centering
    \includegraphics[width=0.83\textwidth]{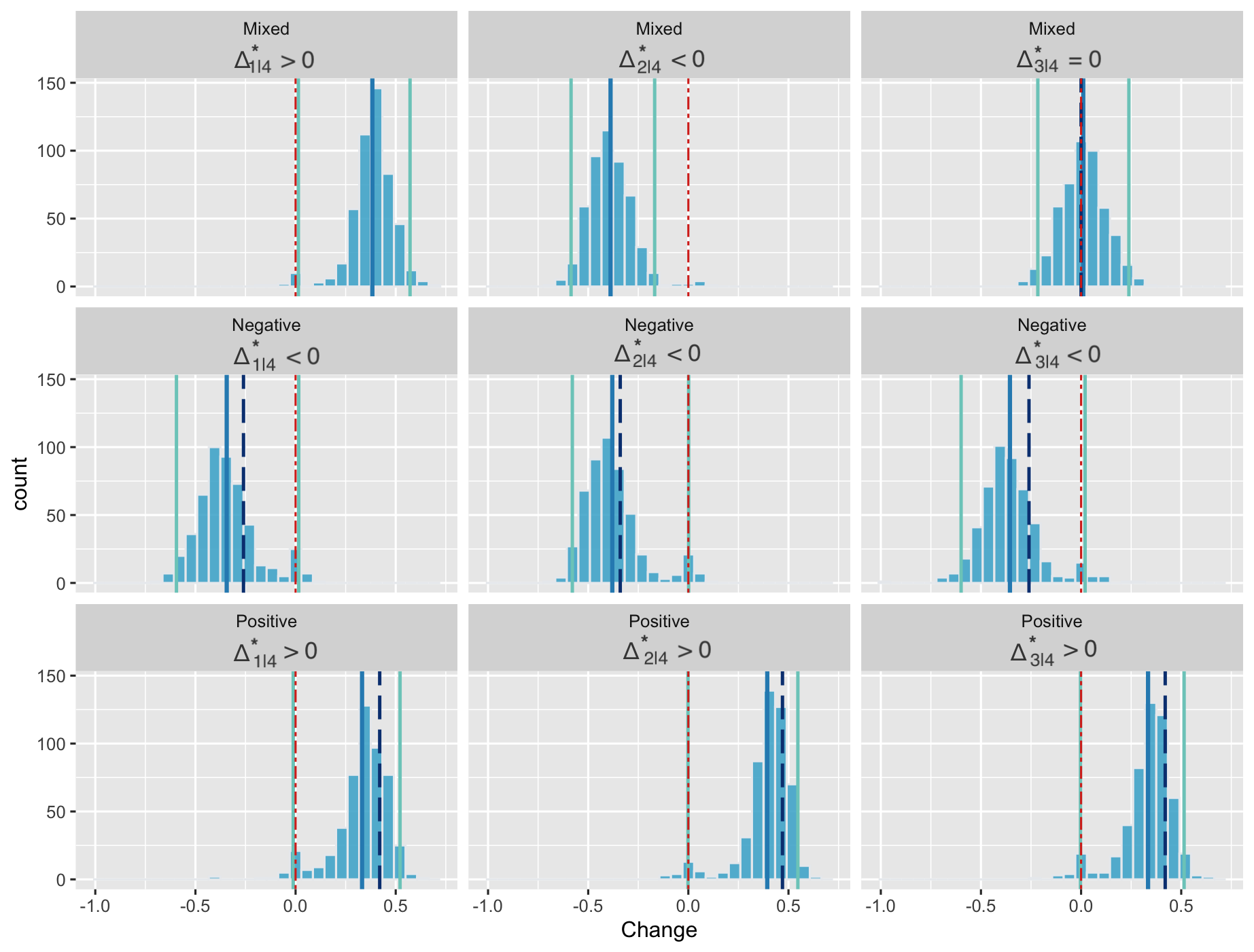}
    \includegraphics[width=0.83\textwidth]{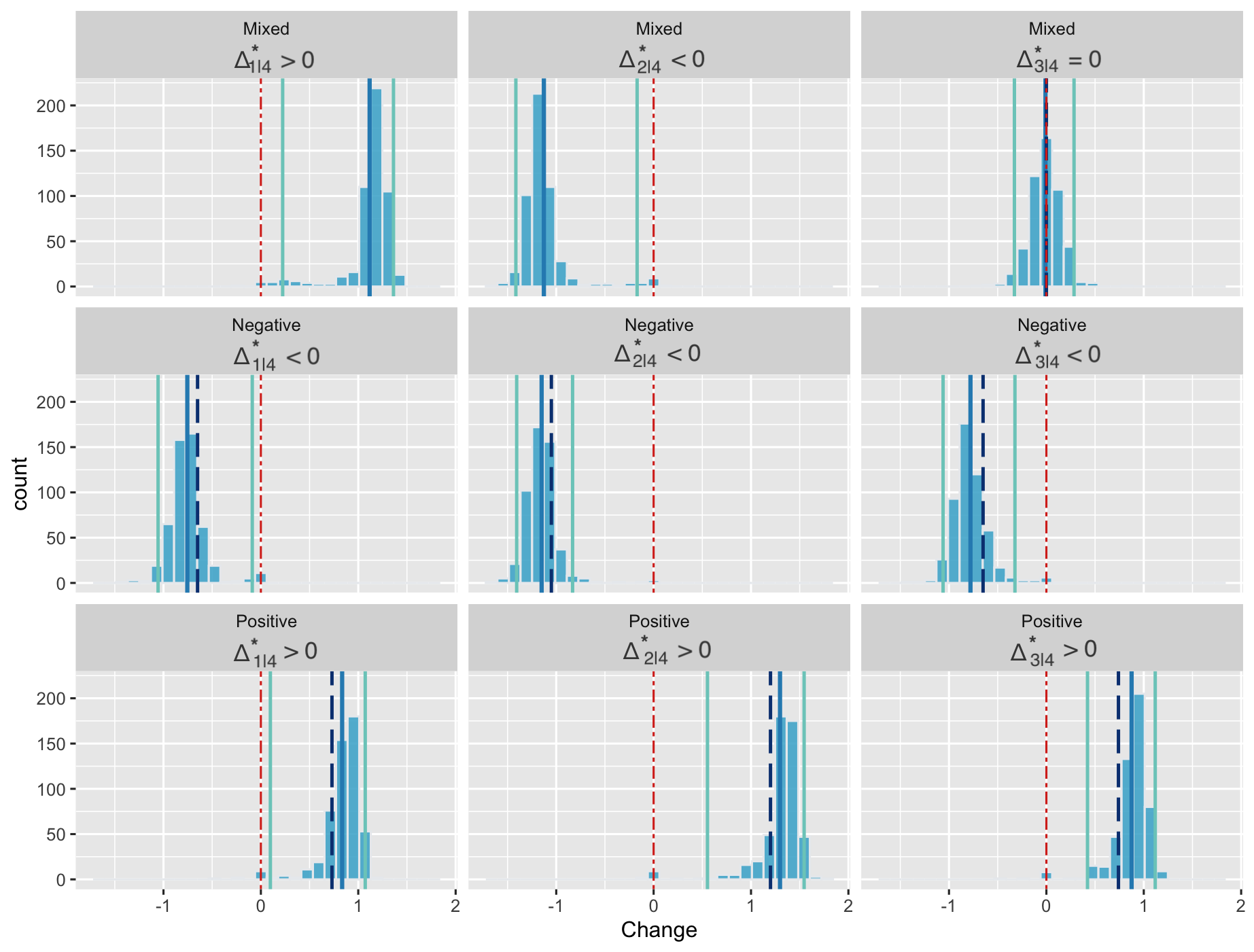}
    \caption{Simulation study: empirical distribution of estimates of $\Delta^{*}_{j|q} = \alpha^{(1)}_{j|q} - \alpha^{(0)}_{j|q}$, $q=4$, $j = 1,2,3$ (left to right) for three scenarios---with relatively low changes (first three rows) and relatively high changes (last three rows). The blue-solid vertical line shows the mean estimate; the light-blue-solid vertical lines are the 95\% corresponding confidence intervals, 
    the dashed-blue vertical line shows the simulated true values; 
    and the red vertical line is at zero.
    }
    \label{SimDelta}
\end{figure}

\subsection{Simulation results} \label{subsec:SimResult}

To show estimation accuracy, we simulate a $p=4$ dimensional $\bs{X}(t)$ from the true model defined in Proposition~\ref{Prop:Simulation} and in \eqref{rawsim}, with three settings for $\Delta_{j|q}$. The settings are: (i) ``positive'' ($\Delta_{j|q} > 0$ for all $j$); (ii) ``negative'' ($\Delta_{j|q} < 0$ for all $j$); (iii) ``mixed'' ($\Delta_{j|q} \in \mathbb{R}^3$, for all $j$). We also test the estimation performance when the change is relatively low (i.e., $\Delta_{j|q} \leq 0.5$) and relatively high (i.e., $\Delta_{j|q} > 0.5$); this gives us a total of six simulation settings. The values for $\Delta_{j|q}$ are selected in such a way that the correlation matrix of $\tilde{\bs{C}}_{\ell}(b)$ in Proposition~\ref{Prop:Simulation} is symmetric positive semi-definite. 
We simulate cluster allocations (i.e., $D_b = d$) uniformly-at-random with probability 0.35 for all $b = 1, \dots, B$. This probability reflects the observed frequency of burst-phase in neonates (see discussion in Section~\ref{subsec:clusteringResult}). We generate $B=1000$ blocks, each with fixed block length of $L = 512$, resulting in a total of $512,000$ time points. We then apply our SpExCon method with a threshold exceedance probability of $\kappa = 0.9$ in Equation~\eqref{kappa}. The experiment is repeated 1000 times and the results are summarized in Figure~\ref{SimDelta}, which shows the distribution of estimates for the three-dimensional vector $\Delta^{*}_{j|q} = \alpha^{(1)}_{j|q} - \alpha^{(0)}_{j|q}$. For all cases, the 95\% bootstrap confidence interval includes of the true value of $\Delta^{*}_{j|q}$. Moreover, we observe that, for the case of high values of true $\Delta^{*}_{j|q}$, the estimates are farther away from zero.
For labeling of burst and non-burst in this simulation study, the spectral clustering algorithm detailed in the Supplementary Material, Section~\ref{appn2}, achieves an average test accuracy of $99.98\%$ and an overall accuracy of $99.82\%$ (training and test sets combined).


\section{EEG Analysis}\label{chap:analysis}

\subsection{Overview}
Recall that the data in our study are from neonates admitted to the intensive care unit of the BAby Brain Activity Center (BABA Center), Finland.
We consider neonates who did and did not experience seizures.
We limit the inclusion criteria of this study to neonates who are older than 38 weeks and who have primary localization of seizure only in the right hemisphere of the brain. As our goal is to identify differences in the spectral tail dependence of stable EEG signals of epileptic and non-epileptic patients, we considered neonates with at least 1000 seconds of seizure-free observations. Note that we only include pre-seizure observations in our analysis. Given these criteria, we include two epileptic and three non-epileptic neonates in our analysis. The structure of the rest of this section is as follows. Section~\ref{subsec:clusteringResult} gives the results of spectral clustering for identifying burst and non-burst phases. Section~\ref{subsec:Estimates} discuss the SpExCon results for epileptic (Section~\ref{subsubsec:ictal}) and for non-epileptic neonates (Section~\ref{subsubsec:nonictal}). 

\begin{figure}
    \centering
    \includegraphics[width=0.99\textwidth]{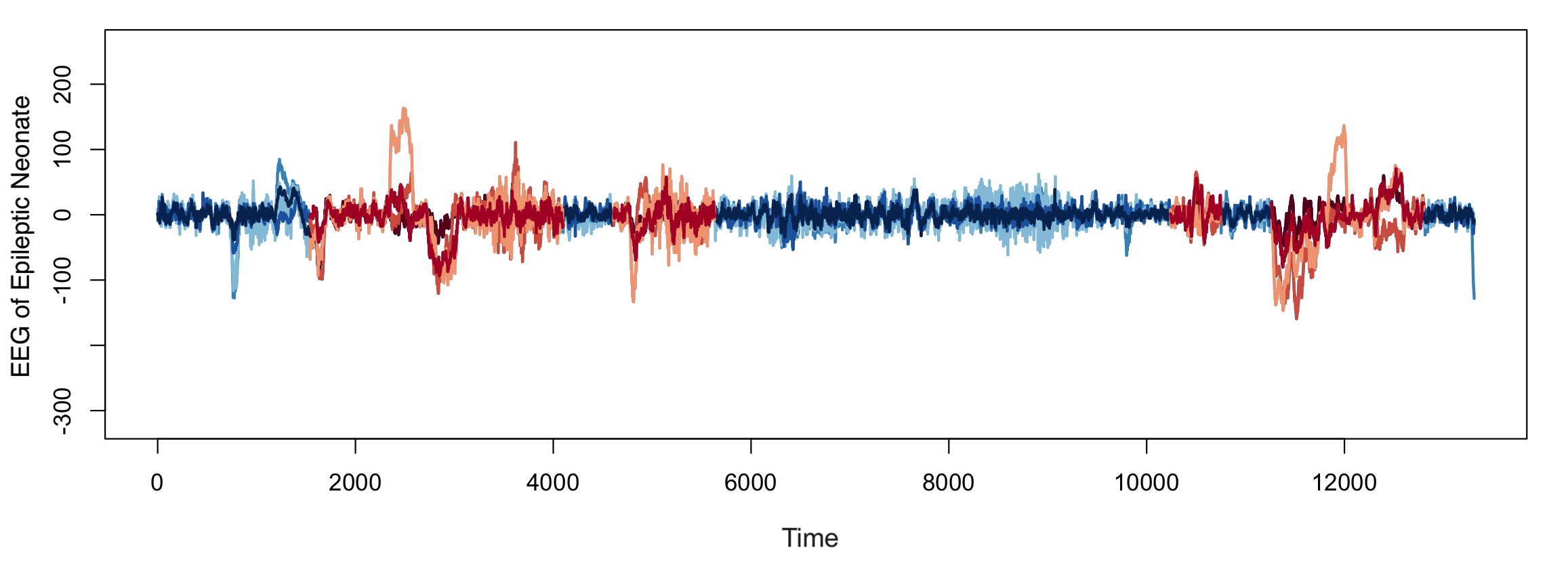}
    \includegraphics[width=0.99\textwidth]{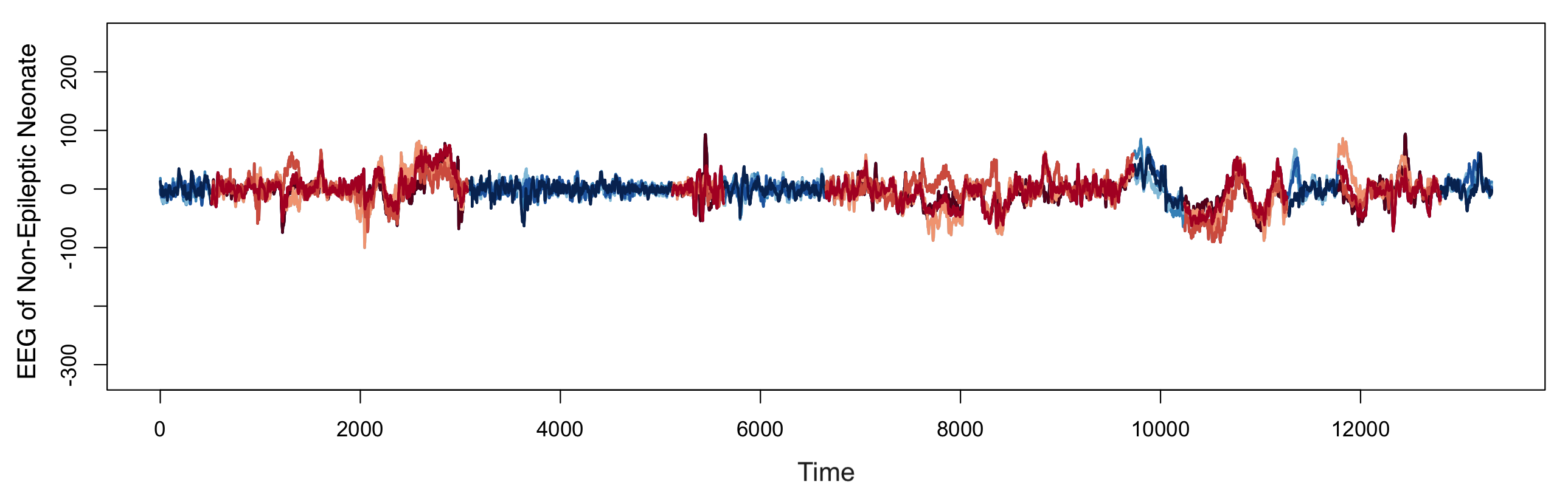}
    \caption{Electroencephalogram at temporal lobe of an epileptic patient (top) and non-epileptic patient (bottom). Red signals, flagged by the spectral clustering algorithm described in the Supplementary Material, Section~\ref{appn2}, are in burst-phase.}
    \label{EEGsubj31}
\end{figure}

\subsection{Spectral Clustering} \label{subsec:clusteringResult}

We first apply the spectral clustering described in \cite{Narula2021} (see the Supplementary Material, Section~\ref{appn2}, for details) to identify burst and non-burst phases. After clustering, we obtain a series of binary indicators $\{D_b\}_{b = 1}^B$ of length $B = 1000$. Figure~\ref{EEGsubj31} shows a sample result of the clustering where signal blocks are classified as burst-suppression ($D_b=1$) or non-burst-suppression ($D_b=0$). For the five neonates, the proportion of burst phases 
are 13.8\%, 28.4\%, 34.2\%, 37.1\%, and 41.8\%. We expect the proportion of burst-phase blocks to be lower, as it occurs less frequently than the non-burst phase \citep{Narula2021}. 
Alongside identification of the phases, we obtain the block-wise periodogram, defined in Equation~\eqref{TVPerio2}, using two-second moving epochs (i.e., $L=512$).

\begin{figure}
    \centering
    \includegraphics[width=0.9\textwidth]{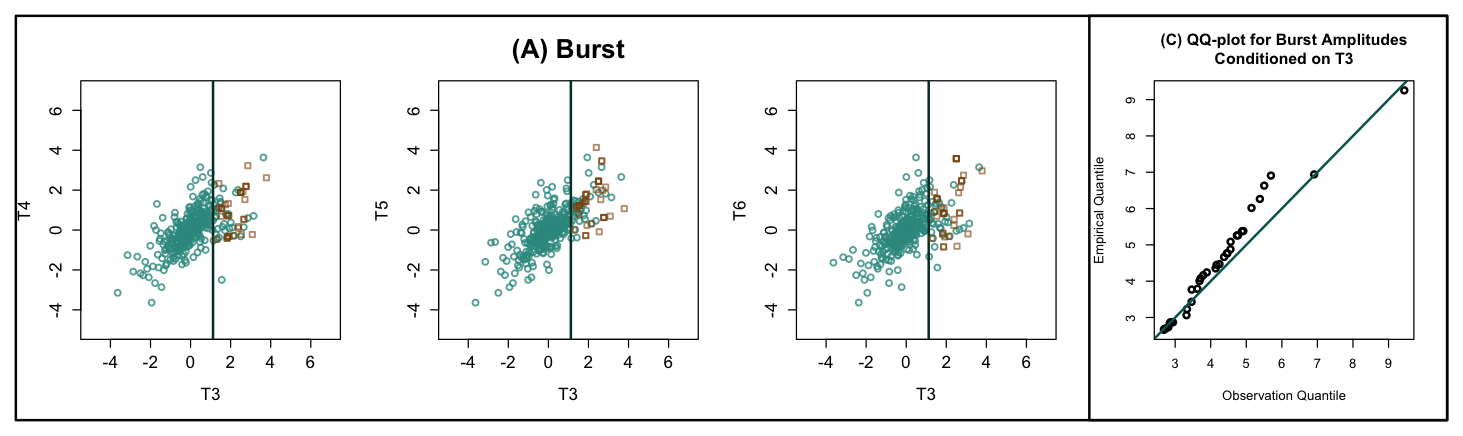}
    \includegraphics[width=0.9\textwidth]{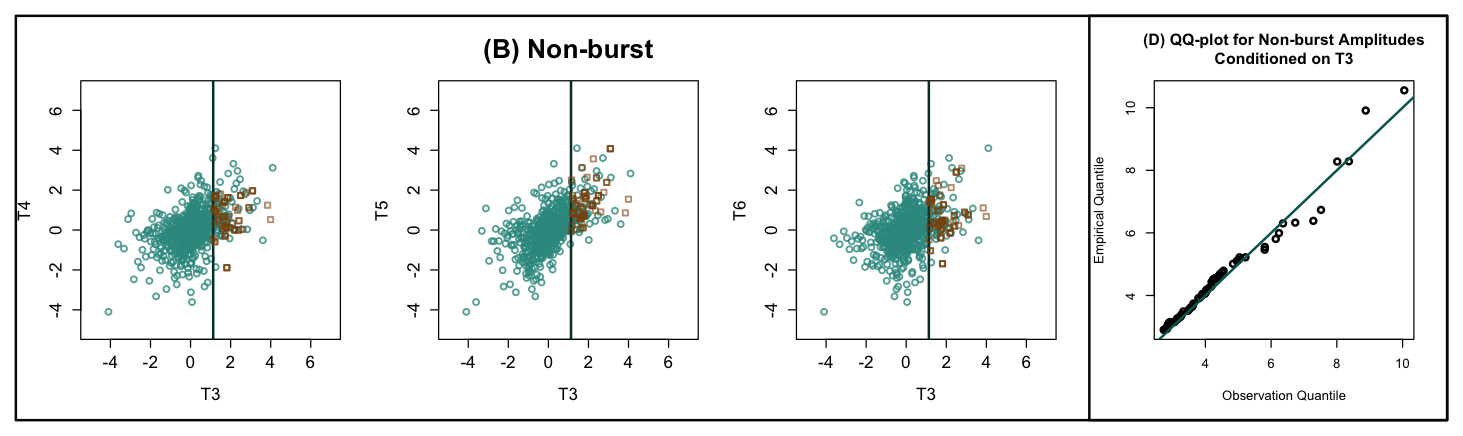}
    \caption{Goodness-of-fit diagnostics for SpExCon estimates. Panels (A) and (B) give scatterplots of simulated (brown squares) and observed (blue circles) values---both transformed to Laplace margins---of the block-wise auto-periodograms for the Delta-band of an epileptic neonate. Solid lines in (A) and (B) separate the upper 10\% and lower 90\% of the observed data. The QQ-plots in (C) and (D) compare the aggregate sum of the observed and simulated auto-periodograms conditioned on channel T3 being large. }
    \label{DeltaFitT5}
\end{figure}

\subsection{SpExCon Estimates} \label{subsec:Estimates}

We next model extreme auto-periodograms for each frequency bands; we use the model defined in Equation~\eqref{DEC} with $\kappa = 0.9$ in Equation~\eqref{kappa}.
For conciseness, we present only results for the delta band, as it is highly linked to seizures and burst-suppression patterns \citep[]{Muthuswamy1999DeltaTheta, Douglass2002BSinNeo}. We plot model diagnostics for channels T4, T5, and T6, conditioned on large values of channel T3 (i.e., $\kappa = 0.9$) for one epileptic and one non-epileptic neonate (shown in Figure~\ref{DeltaFitT5}) but present the parameter estimates and significant results for all five neonates (see Figure~\ref{MDT-HM} and \ref{Kappa}). The rest of the model diagnostics are provided in the Supplementary Material, Section~\ref{appn3}.  

\subsubsection{Epileptic Neonates} \label{subsubsec:ictal}

We consider two epileptic neonates, that we number $1$ and $2$. The first considered neonate that falls within our limitation criteria was diagnosed with hypoplastic left heart syndrome. Neuroimaging data showed that there was an occurrence of right side intraventricular hemorrhage. Within a few minutes of the start of the monitoring, the patient suffered seizure localized in the right hemisphere of the brain. A sample of signals from this neonate is plotted in Figure~\ref{EEGsubj31}. 
The other considered epileptic neonate is around 43--44 weeks old and had a seizure in the occipital-parietal region of her brain. She was diagnosed with neonatal convulsions and shows right-side infarction in the brain. 
We applied the SpExCon model to the periodograms, as defined in Equation~\eqref{TVPerio2}, for Delta-band frequencies. Specifically, we analyze signals from channels T3, T4, T5, and T6 for the first neonate and channels O1, O2, P3, and P4 for the second neonate (as the occipital-parietal region is the seizure focal point for this neonate). Figure~\ref{DeltaFitT5} shows scatter plots of the transformed auto-periodograms of channels T4, T5 and T6 plotted against T3, alongside QQ-plots of simulated aggregates (i.e., based on the fitted model) 
versus observed aggregates ($R_b^{(-q)}$ in Equation~\eqref{aggR}), conditioned on large values of the $q$-th component (channel T3). We observe good fits to these data, as the quantiles of the aggregated simulated and aggregated observed tail values align closely along a straight diagonal line (see the Supplementary Material, Section~\ref{appn3}, for the complete illustration of the goodness-of-fit). Parameter estimates are given in Figure \ref{MDT-HM} and inference on the $\Delta$ parameter is displayed in Figure \ref{Kappa}. 

\begin{figure}
    \centering
    \includegraphics[width=1.02\textwidth]{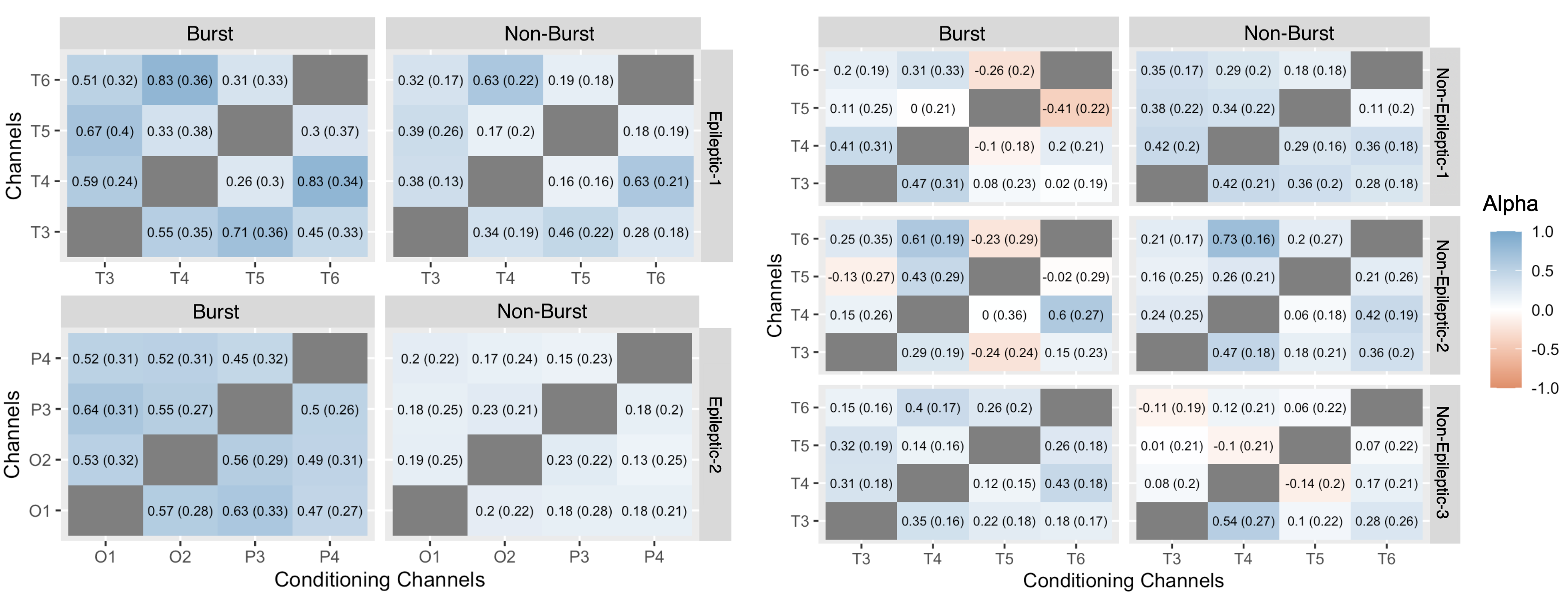}
    \caption{Spectral extremal connectivity estimates of $\alpha^{(d)}_{j|q}$'s (see Equation~\eqref{psi}) and its standard error, enclosed in parentheses, for Delta-band of (left block) epileptic patients (i.e., those who have had seizure) and (right block) non-epileptic patients (i.e., with normal brain imaging findings). The columns of each heatmap corresponds to the conditioning channel $q$, while the rows are the other channels. Results are reported for the two phases, namely non-burst ($d=0$) and burst ($d=1$); see the panel titles.}
    \label{MDT-HM}
    \includegraphics[width=0.99\textwidth]{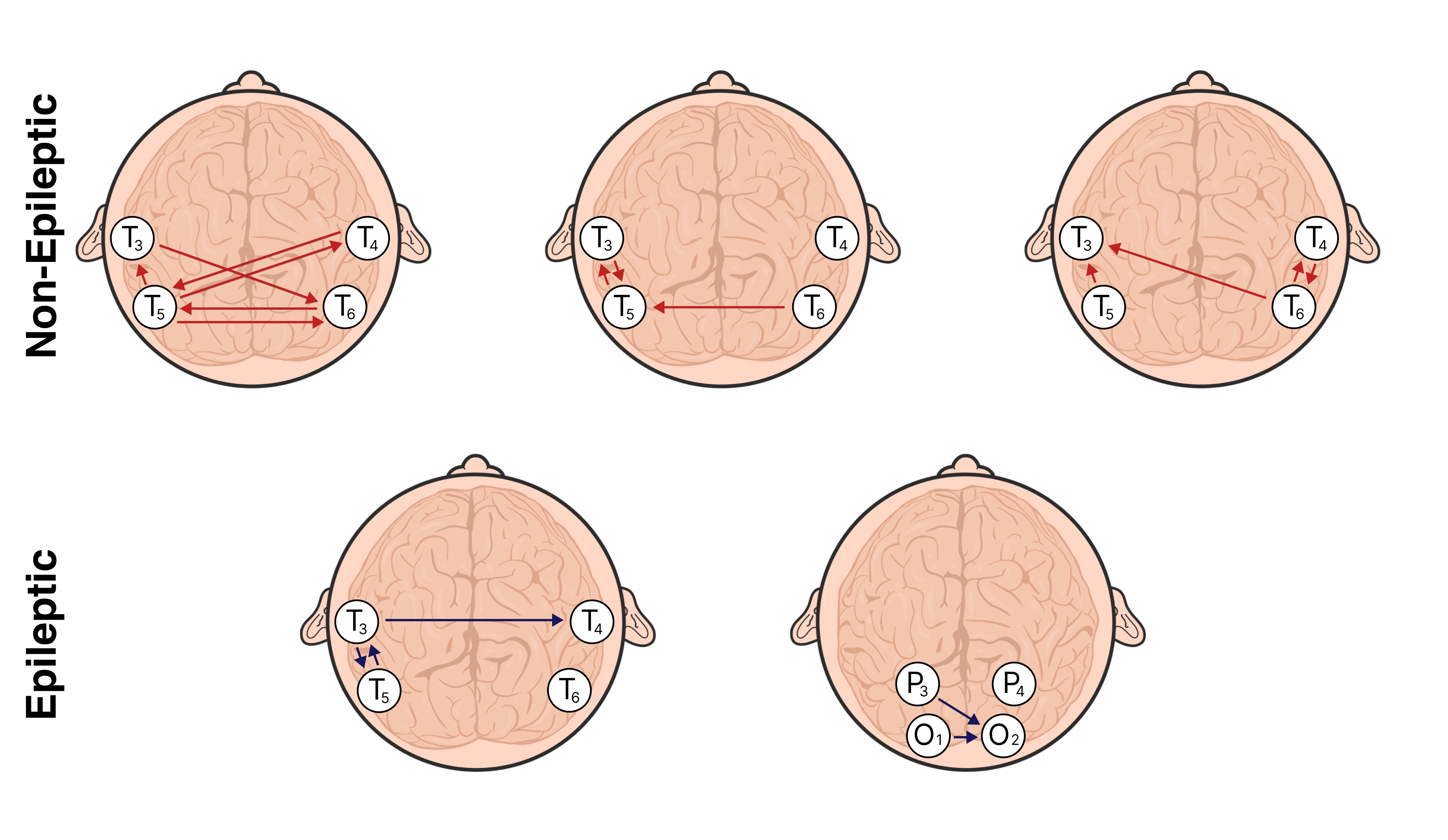}
    \caption{Significant estimates of $\Delta_{j|q}$, corresponding to the arrow $q \rightarrow j$, indicate the direction of significant changes in the connectivity between channels $q$ and $j$ (conditioned on high values at channel $q$).  
    Red arrows correspond to negative $\Delta_{j|q}$ (burst-phase has lower extremal connectivity) and blue corresponds to positive $\Delta_{j|q}$. The first row is for the three non-epileptic patients and the second row is for the two epileptic patients.}
    \label{Kappa}
\end{figure}

For both epileptic patients, we observe positive extremal associations across all channels in both burst and non-burst phases (see Figure~\ref{MDT-HM}). Furthermore, all significant changes in the epileptic group, as shown in Figure~\ref{Kappa}, are positive, suggesting increased extremal brain connectivity when in burst-phase.
Notably, the strongest extremal dependence in epileptic patient 1 occurred between channels T6 and T4. Additionally, the extremal association between T3 and T4 varied in magnitude between phases, with T4 showing a stronger positive association with T3 during the burst phase compared to the non-burst phase. A similar pattern is observed between T3 and T5, where high amplitudes in T3 were more likely to coincide with high amplitudes in T4 and T5 during the burst-suppression phase.
For epileptic patient 2, signals from channels P3 and O1 exhibit the strongest extremal dependence during the burst phase. However, this level of dependence was not observed in the non-burst phase of the same patient.
Significant values of positive $\Delta$ suggest that there is higher tail-connectivity during the burst-phase compared to the non-burst phase (see Figure \ref{Kappa}). 
Channel T3 (of epileptic patient 1) and P3 (of epileptic patient 2) may play a pivotal role in understanding the epileptic activity of the two neonates, as their high delta amplitudes could heighten activity in the right hemisphere where the seizure occurred (see Figure \ref{Kappa}). This finding aligns with the definition of \cite{NIH23}, that the occurrence of seizures stem from outbursts of brain energy of high oscillations.

\subsubsection{Non-epileptic Neonates} \label{subsubsec:nonictal}

We now analyze EEG data from a neonate (around 38--39 weeks old) NICU patient. 
The EEG showed no occurrence of seizure throughout the monitoring.
We also analyze two more non-epileptic neonates who were both around 42--43 weeks old during the time of recording. 

The results for these three non-epileptic neonates are summarized in Figures~\ref{MDT-HM} and \ref{Kappa}. Significant values of $\Delta_{j|q}$ estimates indicate lower tail dependence during the burst-phase, conversely to what we observed for epileptic neonates. 
This lowered tail connectivity in the burst-suppression phase of non-epileptic neonates might support a well-functioning inhibitory neurotransmitter theory.  
Our method could potentially be used in clinical trials to investigate negative extremal associations in healthy patients and to assess the impact of treatment on patients. Thus, our proposal could provide a theoretical framework to understand the desynchronization and neurotransmitter theory which are the basis of seizure management. Precisely, desynchronization theory suggests that treatments can stimulate an anti-epileptic effect by restoring balance to excessively ``synchronized'' neuronal activity, while neurotransmitter theory suggests that after treatment there is an increase in inhibitory neurotransmitters hindering transmission of neuronal communication \citep{ogbonnaya2013vagal}.




\section{Conclusion}\label{chap:conclusion}
One of most common and well-established brain connectivity measures is coherence, which reveals spectral association in the bulk of the distribution. 
While coherence is informative in terms of the linear association between two signals at a specific frequency, it does not take into account the association that is only highlighted in the tail of the joint distribution. In studying seizures, clinical and quantitative studies have shown linkage of seizures to excessive increase in electrical activity from multiple neurons \citep{NIH23}. Motivated by these findings, we developed a novel statistical model, SpExCon, to investigate spectral association between channels that focuses on extremely high amplitude (or spectral power) at a specific frequency band. The proposed SpExCon method aims to identify pivotal brain-connectivity related to burst-suppression patterns and seizures. 

SpExCon provides a novel framework for detecting and quantifying changes in extremal dependence across brain regions. Our method's ability to identify frequency-specific tail dependencies offers several key advantages for neuroscience applications. Notably, it enables the detection of extreme synchronization patterns that may serve as precursors to seizure events, as demonstrated by the distinct tail dependence profiles observed between epileptic and non-epileptic events. This capacity for differentiation could potentially assist in seizure prediction and intervention strategies.

Our work uses a spectral conditional tail dependence model, motivated by the conditional modeling framework of \cite{Heffernan2004}, with tail dependence parameters changing with the state of the process. Simulation results have shown accurate estimation of the parameters and the direction of the association. One advantage of our method is that it provides estimates of spectral tail dependence measure at specific frequency bands, which highlights differences in brain connectivity for epileptic and non-epileptic neonates. 
Our method detects differences in tail-association of EEG signals during the burst-suppression and regular phases of non-epileptic and epileptic neonates.
One of our novel findings is a higher frequency-specific tail dependence observed during the burst-suppression of epileptic patients whereas non-epileptic patients show lower frequency-specific tail dependence during the burst-suppression phases.
These novel findings facilitate a deeper understanding of functional brain connectivity during extreme events, e.g., seizures. Our model holds versatile applications, such as environmental, where different rainfall types may have varying extremal dependence structures \citep[see, e.g.,][]{richards2023joint}.


SpExCon's conditioning approach provides methodological advantages over existing extremal dependence models for brain connectivity (Guerrero et al., 2023; Redondo et al., 2024) through its incorporation of information about reference states in its vector parameters. By explicitly modeling the transition between baseline and extreme states, our method can detect subtle changes in dependence structures that might be masked by traditional coherence analyses. This state-dependent modeling approach is particularly valuable in clinical settings where patient-specific baseline variations must be considered.

While our primary application focused on seizure analysis in the 0-50 Hz interval, SpExCon's framework is generalizable to other applications requiring the analysis of extremal dependencies in multivariate time series as it could be used to study extreme synchronization events in other pathological conditions or during cognitive tasks. Furthermore, the method could be extended to analyze synchronization patterns in higher frequency bands (80-500 Hz) in intracranial EEG, which have been shown to be potential predictors of post-surgical outcomes in epilepsy patients \citep{PintoOrellana2024}. This extension to higher frequencies could provide new insights into the relationship between extreme spectral dependencies and treatment outcomes.

\begingroup

\if1\blind
{


} \fi

\endgroup

\baselineskip=14pt
\begingroup
\bibliographystyle{apalike}
\bibliography{bibliography}

\newpage
\setstretch{1.5}


\pagenumbering{arabic}
\renewcommand*{\thepage}{\arabic{page}}
\counterwithin{figure}{section}
\renewcommand{\thesection}{A-\arabic{section}}
\setcounter{section}{0}

\bigskip
\begin{center}
{\large\bf SUPPLEMENTARY MATERIAL}
\end{center}

This supplementary document provides additional details to support the methodology and results presented in the main manuscript. Section~\ref{appn2} describes the spectral clustering algorithm used for burst-suppression identification adopted from \cite{Narula2021}. Section~\ref{appn1} presents the proof of Proposition 4.1., detailing the mathematical derivations and justifications. Finally, Section~\ref{appn3} includes the full set of QQ plots used to assess the goodness-of-fit of the SpExCon model on the data described in Section~\ref{chap:analysis}.

\section{Spectral Clustering for Burst-Suppression Identification} \label{appn2}

The second step towards our goal of measuring dependence between the EEG channels through their variance (energy) contributions is to label (or cluster) the different time blocks $b$ as either burst or non-burst.
This will consequently lead to identifying differences in the connectivity structure between the burst and non-burst phases. 
A new unsupervised machine learning technique through spectral clustering was proposed in \cite{Narula2021}, which we adopt here. In collaboration with neurologists, \cite{Narula2021} label the burst-suppression pattern in the data, which served as the reference, and compared it with the results of their spectral clustering method. Spectral clustering demonstrated competitive performance against deep learning techniques and exhibited high accuracy.

To simplify the notation, we will use ${\bm X}^{b}(t) := \{{\bm X}(t)\}_{t \in \mathcal{T}_b}$ to be the filtered EEG for block $b$. 
Let {$\bs{C}_b$} be the covariance matrix of the filtered signals ${\bm X}^{b}(t)$ at block $b$ that we restrict to always be positive semi-definite to ensure that eigenvalues are positive.
Let $\{\lambda_{j,b,b'}\}_{j = 1}^p$ be the ordered solutions to the equation ${\bs C}_b{\bs y} = \lambda{\bs C}_{b'}{\bs y}$ for blocks $b \neq b'$, and $\{\bs{y}_j\}_{j = 1}^p$ the corresponding eigenvectors. Intuitively, $\lambda_{j,b,b'}$ carries the information on the ``ratio" between ${\bs C}_b$ and ${\bs C}_{b'}$. Hence, all values of $\{\log(\lambda_{j,b,b'})\}_{j = 1}^p$ provide together a suitable measure of distance between the matrices ${\bs C}_b$ and ${\bs C}_{b'}$. Here, dissimilarity in the underlying processes that generated the EEG signals in blocks $b$ and $b'$ will be characterized by $g(b,b') = ||\{\log (\lambda_{j,b,b'})\}_{j = 1}^p||_2$ \citep{Forstner2003}.
Note that when the covariance matrices for blocks $b$ and $b'$ are identical, then $\lambda_{j,b,b'} = 1$ for all $j$ and, hence, $g(b,b') = 0$. 
Now, we set up the similarity matrix between all blocks ($b=1, \ldots, B$) to be $\bs{S}$ (which is $B \times B$) whose elements $S_{b,b'}$ are given by
\[ S_{b,b'} = \exp\left\{ - \frac{g(b,b')^2}{2\eta}\right\}, \]
where $\eta$ is the median of $\big\{ g(b,b') \big\}$ across all pairs $(b,b')$.


We then use the eigenvectors, $\{\bs{y}_j\}_{j = 1}^p$ where $\bs{y}_j = (y_{j,1}, \dots, y_{j,B})^\top$, of the similarity matrix, $\bs{S}$, to cluster the blocks of multivariate time series into two clusters. This is obtained by solving another generalized eigenvalue system given by $(\bs{V-S})\bs{y} = \bs{\lambda Vy}$ where $\bs{V}$ is the degree matrix of $\bs{S}$ \citep{Narula2021} and then conducting $k$-means clustering, $k = 2$, of the first two eigenvectors which correspond to the first two largest eigenvalues. Details of spectral clustering can be found in \cite{Ng2001Sclust}. The output in this step is the binary labels $D_b$ where $D_b = 1$ if the EEG in block $b$ belongs to the burst phase and $0$ if it is a non-burst phase. To distinguish the clusters as burst or non-burst, we calculate the average total variance of the clusters. The cluster of blocks that has a lower probability of occurrence and has a higher variance (on the average) is the cluster of blocks that we label as burst. 
We now denote the set of blocks that belong to the non-burst and burst phases to be, respectively,  $\mathcal{B}_0$ and $\mathcal{B}_1$, i.e., $\mathcal{B}_0 = \{b: D_b = 0\}$ and 
$\mathcal{B}_1 = \{b: D_b = 1\}$. 

\paragraph{Accuracy and homogeneity of spectral clustering.} To check the accuracy of the spectral clustering approach, we use a 70\%--30\% separation of training and test data. The first 70\% of the $B$ blocks are used as training data ($b^* = 1, \dots, \lfloor 0.7B \rfloor$) and the remaining used for testing ($b^{\Sc} = 1, \dots, \lfloor 0.3B \rfloor$). Spectral clustering is performed on the training set, and the similarity distance ($S_{b^\Sc b^*}$) of training blocks with respect to the test blocks is computed. For each block $b^\Sc$ in the test set, we then take the five nearest neighbors in the training set, i.e., the blocks $b^*$ with the five highest values of $S_{b^{\Sc}b^*}$. The label that is the most frequent among the nearest neighbors is used as the predicted label for block $b^\Sc$. In the simulation study detailed in Section~\ref{chap:simulation}, we report the average sensitivity (or true positive rate) and specificity (true negative rate) in the test set of this method. 


\section{Proof of Proposition~\ref{Prop:Simulation}} \label{appn1}

Here we detail the proof of Proposition 4.1.
\begin{duplicate} 
    Consider a `locally stationary' $p$-dimensional time series $\{\bs{X}(t)\}_{t=1}^T$ so that the time series is approximately stationary for distinct blocks of length $L$, indexed by time points $\mathcal{T}_b := \{(b-1)L+1 , \dots , bL\}, b = 1, \dots, B$, such that $T = LB$. Specifically, let ${\bs{X}(t) = \bs{C}(t)\bs{O}(t) + \bs{W}(t),}$ for $t \in \{\mathcal{T}_b\}_{b = 1}^B$, where $\bs{C}(t) \in \mathbb{R}^{p \times 5}$ is assumed to be constant within a time block $\mathcal{T}_b$, i.e., $\bs{C}((b-1)L+1) = \dots = \bs{C}(bL)$, for $b \in \mathcal{B}_d$, and $d = 0, 1$, such that $\mathcal{B}_0 \cap \mathcal{B}_1 = \emptyset$ and $\mathcal{B}_0 \cup \mathcal{B}_1 = \{1, \dots, B\}$. Moreover, assume the vector $\bs{O}(t) = [O_1(t), \ldots, O_5(t)]^\top$ is comprised of mutually-independent AR(2) processes, and $\bs{W}(t) \in \mathbb{R}^{p}$ is an independent vector of white-noise processes.
    Define the vector $\tilde{\bs{C}}_{\ell}(b) := (\tilde{C}_{1\ell}(b), \dots, \tilde{C}_{p\ell}(b))^\top$, for a fixed $\ell \in \{1, \dots, 5\}$ and $b \in \mathcal{B}_d$, such that 
    $\tilde{C}_{j\ell}(b) := \{(C_{j\ell}(t))^2: t \in \mathcal{T}_b\} $, for $j = 1, \dots, p$. 
    Assume that $\{\tilde{\bs{C}}_{\ell}(b)\}_{b \in \mathcal{B}_d}$, for $d = 0, 1$, has a Gaussian copula with exponential margins, and let $\rho^{(\ell,d)}_{qj} > 0$ denote the correlation between $\{\tilde{C}_{q\ell}(b)\}_{b \in \mathcal{B}_d}$ and $\{\tilde{C}_{j\ell}(b)\}_{b \in \mathcal{B}_d}$. Then, for a \underline{fixed $\ell$}, the SpExCon tail dependence parameters are $\alpha^{(d)}_{j|q} = (\rho^{(\ell,d)}_{qj})^2$, $a_{j|q} = \tanh^{-1}\{(\rho^{(\ell,0)}_{qj})^2\}$, $\beta_{j|q} = 1/2$, and $\Delta_{j|q} = \tanh^{-1}\{(\rho^{(\ell,1)}_{qj})^2\} - \tanh^{-1}\{(\rho^{(\ell,0)}_{qj})^2\}$.
\end{duplicate}

\begin{proof}
    Say, for channel $j$ and block $b$, the periodogram (see Section 3.1. Equation (6)) of $\{X_j(t)\}_{t \in \mathcal{T}_b}$ is
\begin{equation*} 
    I_j(b,\omega_r) = \frac{1}{L} \bigg| \sum_{t\in \mathcal{T}_b} \left(C_{j1}(b)O_1(t) + \dots + C_{j5}(b)O_5(t) + W_j(t) \right) \exp(-i2\pi \omega_r (t - [(b-1)L + 1])) \bigg|^2.
\end{equation*}
For a specific frequency band $\Omega_\ell$, the mutual independence of $\bs{O}(t)$ ensures zero frequency outside the band of interest, hence,
\begin{align*}
    I_j(b,\omega_r) &=  \frac{1}{L} \bigg| \sum_{t\in \mathcal{T}_b} \left(C_{j\ell}(b)O_{\ell}(t) + W_j(t) \right) \exp(-i2\pi \omega_r (t - [(b-1)L + 1])) \bigg| ^2  \\
    &= \frac{1}{L} \bigg| \sum_{t\in \mathcal{T}_b} 
    C_{j\ell}(b)O_{\ell}(t) \exp(-i2\pi \omega_r (t - [(b-1)L + 1])) \\ &+
    \sum_{t\in \mathcal{T}_b} 
    W_j(t) \exp(-i2\pi \omega_r (t - [(b-1)L + 1])) \bigg|^2. \notag
\end{align*}
   
Note: $\text{Cov}(O_{\ell}(t), W_{\ell}(t)) = \sum_{t} O_{\ell}(t)W_{\ell}(t) = 0$. Define the periodograms of $O_\ell(t)$, for $\ell = 1,\dots,5$, and $W_j(t)$, for $j = 1,\dots,p$, respectively, as
\begin{align*}
    I^{O}_{\ell}(\omega_r) &= \frac{1}{L} \bigg|\sum_{t \in \mathcal{T}_b} 
    O_{\ell}(t) \exp(-i2\pi \omega_r (t - [(b-1)L + 1])) \bigg|^2, \text{ and } \\
     I^W_{j}(\omega_r) &= \frac{1}{L} \bigg|\sum_{t \in \mathcal{T}_b} 
    W_{j}(t) \exp(-i2\pi \omega_r (t - [(b-1)L + 1])) \bigg|^2.
\end{align*}
Hence,

\begin{align*}
    I_j(b,\omega_r)  &= \frac{1}{L} 
    \Bigg[ \bigg| \sum_{t\in \mathcal{T}_b} 
    C_{j\ell}(b)O_{\ell}(t)\exp(-i2\pi \omega_r (t - [(b-1)L + 1])) \bigg| ^2 \\ &+
    \bigg| \sum_{t\in \mathcal{T}_b} 
    W_j(t)\exp(-i2\pi \omega_r (t - [(b-1)L + 1])) \bigg| ^2 \Bigg] \notag \\
    &= (C_{j\ell}(b))^2 I^{O}_{\ell}(\omega_r) + I^W_{j}(\omega_r).
\end{align*}
Recall that ${Q_{\ell} = \{\omega_r: \ s\omega_r \in [-\Omega_{\ell}^{(2)}, -\Omega_{\ell}^{(1)}) \ \bigcup \ (\Omega_{\ell}^{(1)}, \Omega_{\ell}^{(2)}] \}}$ is the set of fundamental frequencies corresponding to band $\Omega_\ell$. Moreover, let
\[ {I_{O}(\Omega_\ell) = \frac{1}{|Q_\ell|} \sum_{\omega_r \in Q_\ell} I^{O}_{\ell}(\omega_r)} \; \; \text{ and } \; \;  {I^{W}_j(\Omega_\ell) = \frac{1}{|Q_\ell|} \sum_{\omega_r \in Q_\ell} I^W_{j}(\omega_r)}.
\]
Then,
\begin{align*}
    I_j(b,\Omega_\ell) &= \frac{1}{|Q_\ell|} \sum_{\omega_r \in Q_\ell} \left((C_{j\ell}(b))^2 I^{O}_{\ell}(\omega_r) + I^W_{j}(\omega_r) \right) = (C_{j\ell}(b))^2 I_{O}(\Omega_\ell) + I_W(\Omega_\ell). 
\end{align*}

\noindent Since $I_{O}(\Omega_\ell)$ and $I^{W}_j(\Omega_\ell)$ are constant over the block $b$, the tail-dependence of $\bs{I}(b,\Omega_\ell)$ is determined by the tail-dependence of $\tilde{\bs{C}}_{\ell}(b)=(C_{1\ell}^2(b),\dots, C^2_{p\ell}(b))$; thus, if $\tilde{\bs{C}}_{\ell}(b)$ follows a Gaussian copula, results from \cite{Heffernan2004} imply that the SpExCon dependence parameters of the multivariate periodogram, $\bs{I}(b,\Omega_1),$ are $\alpha^{(d)}_{j|q} = (\rho^{(\ell,d)}_{qj})^2$ and ${\beta_{j|q} = 1/2}$ for $j,q=1,\dots,p, j\neq q$. Consequently, we have $\alpha^{(0)}_{j|q} = (\rho^{(\ell,0)}_{qj})^2 = \tanh(a_{j|q})$ and $\Delta_{j|q} = \tanh^{-1}\{(\rho^{(\ell,1)}_{qj})^2\} - \tanh^{-1}\{(\rho^{(\ell,0)}_{qj})^2\}$ for $j,q=1,\dots,p, j\neq q$.

\end{proof}

\section{Goodness-of-fit} \label{appn3}
    The following graphs show the goodness-of-fit diagnostics for SpExCon estimates. Panels (A) and (B) display scatterplots of simulated (brown squares) and observed (blue circles) values (both transformed to Laplace margins) of the block-wise auto-periodograms for the Delta-band of the specified neonate. Solid lines in (A) and (B) separate the upper 10\% and lower 90\% of the observed data. The QQ-plots in (C) and (D) compare the aggregate sum ($R^{(-q)}_b$) of the observed and simulated auto-periodograms, conditioned on channel $q$ being large.

\begin{figure}[h!]
    \centering
    \includegraphics[width=1.01\textwidth]{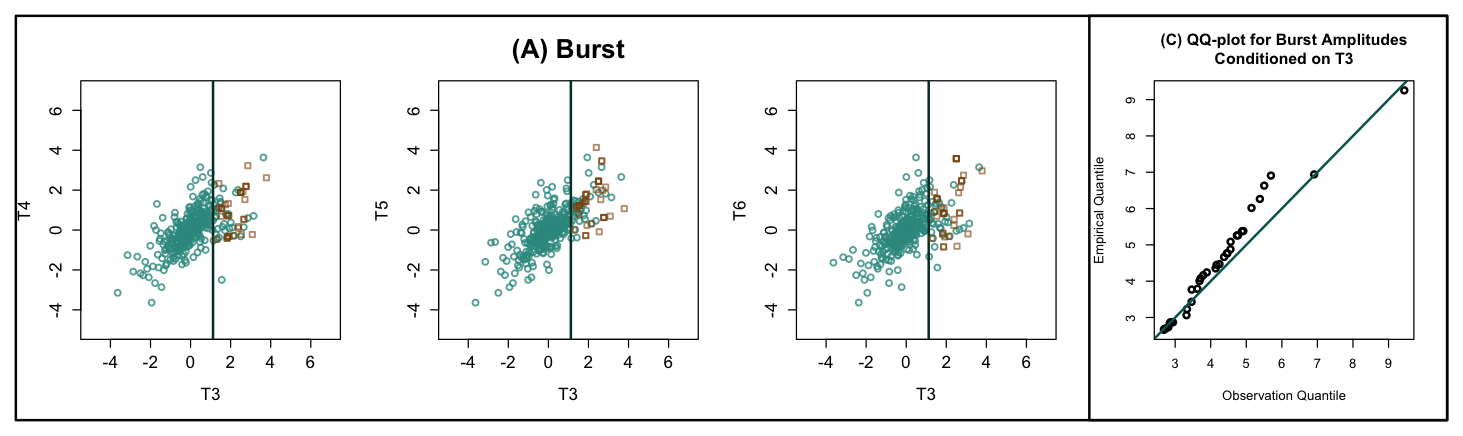}
    \includegraphics[width=1.01\textwidth]{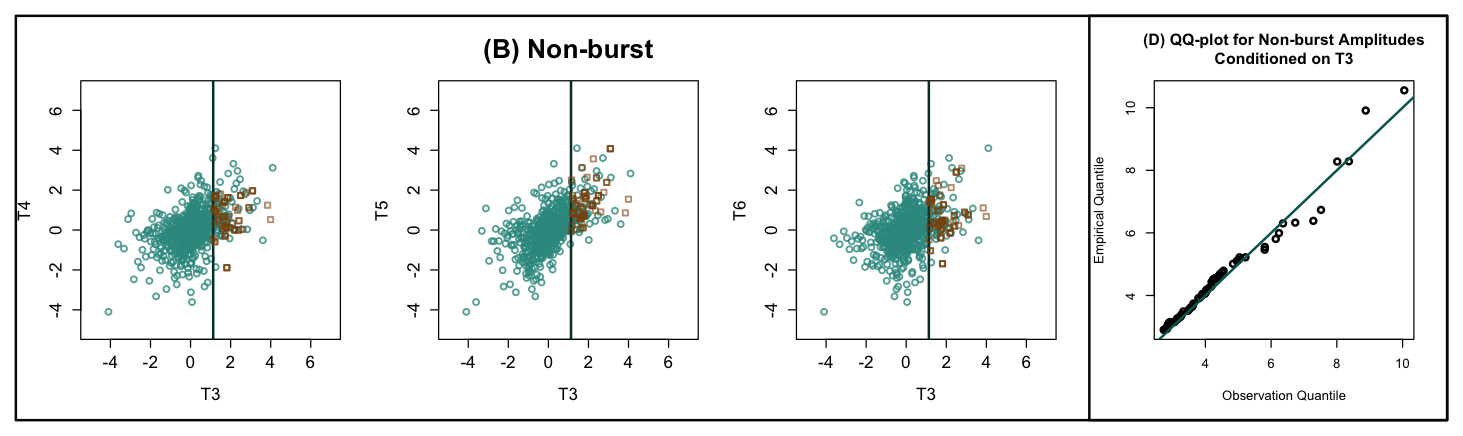}
    \caption{SpExCon estimates diagnostics for \textit{epileptic neonate 1}, conditioned on large values of channel T3.}

\end{figure}

\begin{figure}
    \centering
    \includegraphics[width=1.01\textwidth]{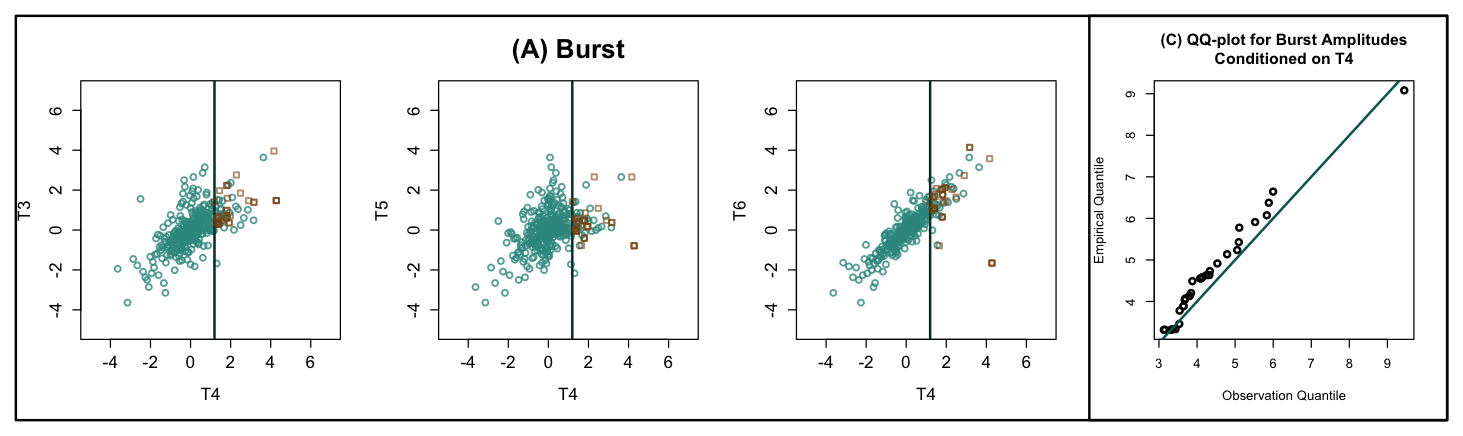}
    \includegraphics[width=1.01\textwidth]{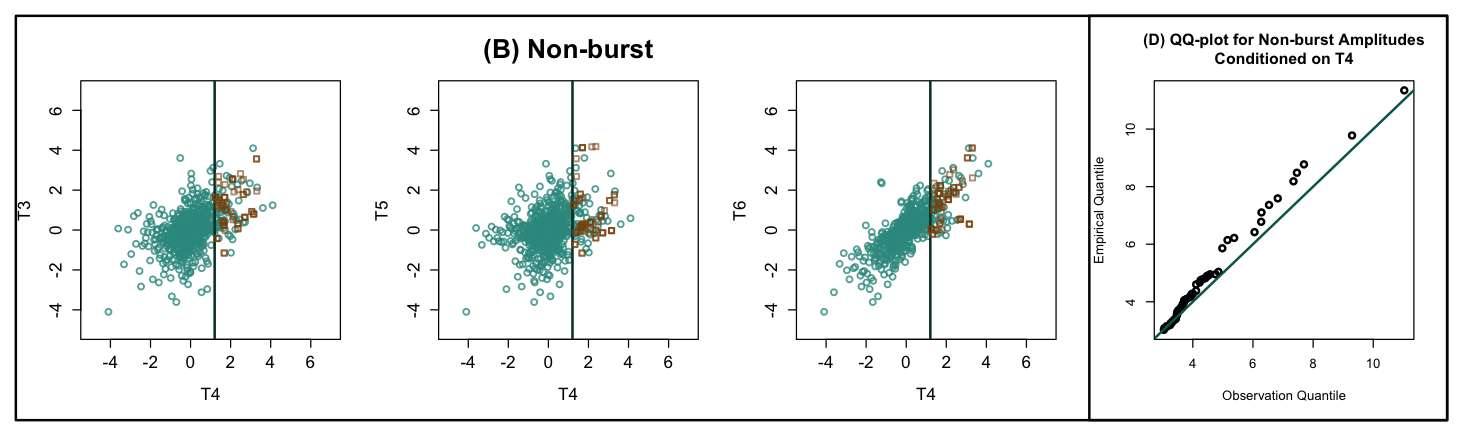}
    \caption{SpExCon estimates diagnostics for \textit{epileptic neonate 1}, conditioned on large values of channel T4.}

    \includegraphics[width=1.01\textwidth]{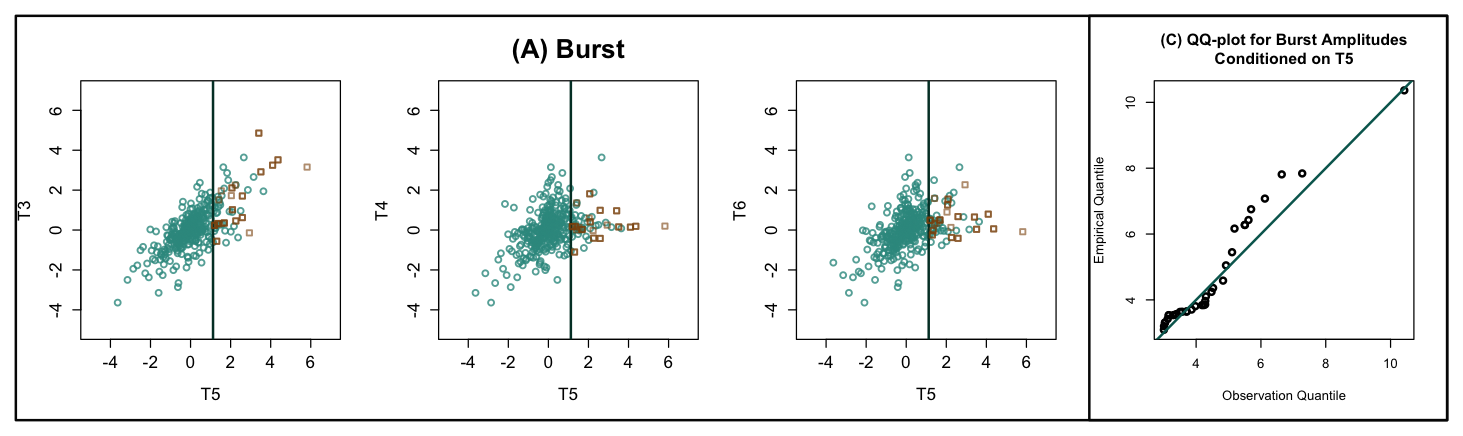}
    \includegraphics[width=1.01\textwidth]{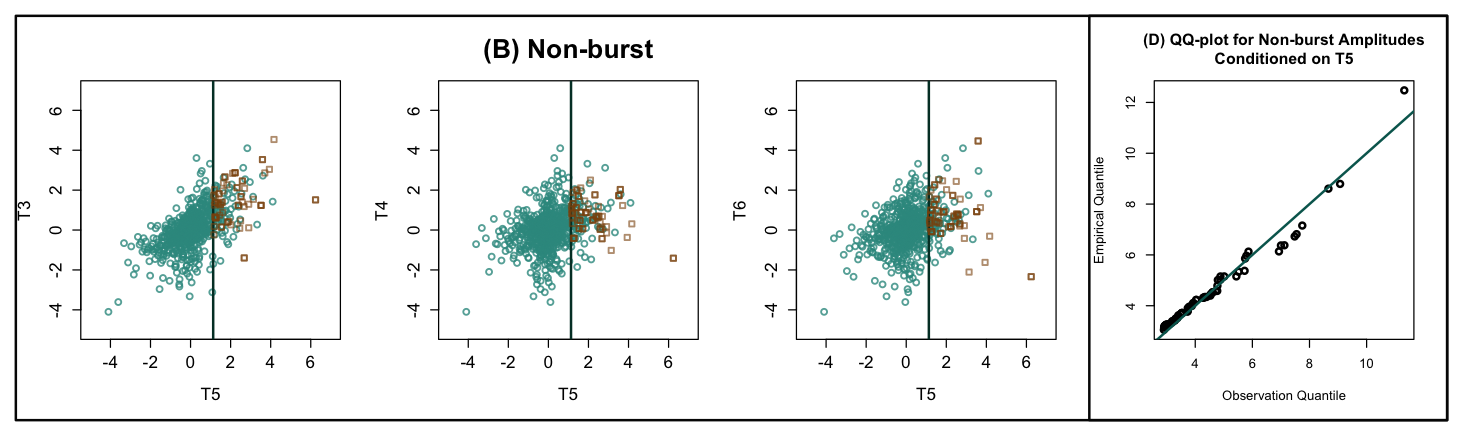}
    \caption{SpExCon estimates diagnostics for \textit{epileptic neonate 1}, conditioned on large values of channel T5.}

\end{figure}

\begin{figure}
    \centering
    \includegraphics[width=1.01\textwidth]{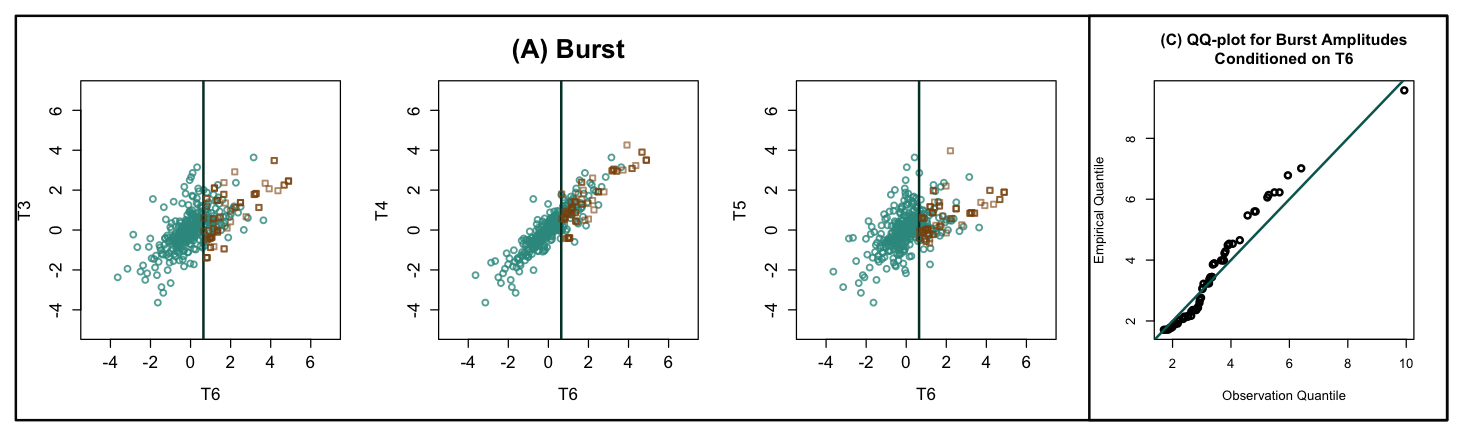}
    \includegraphics[width=1.01\textwidth]{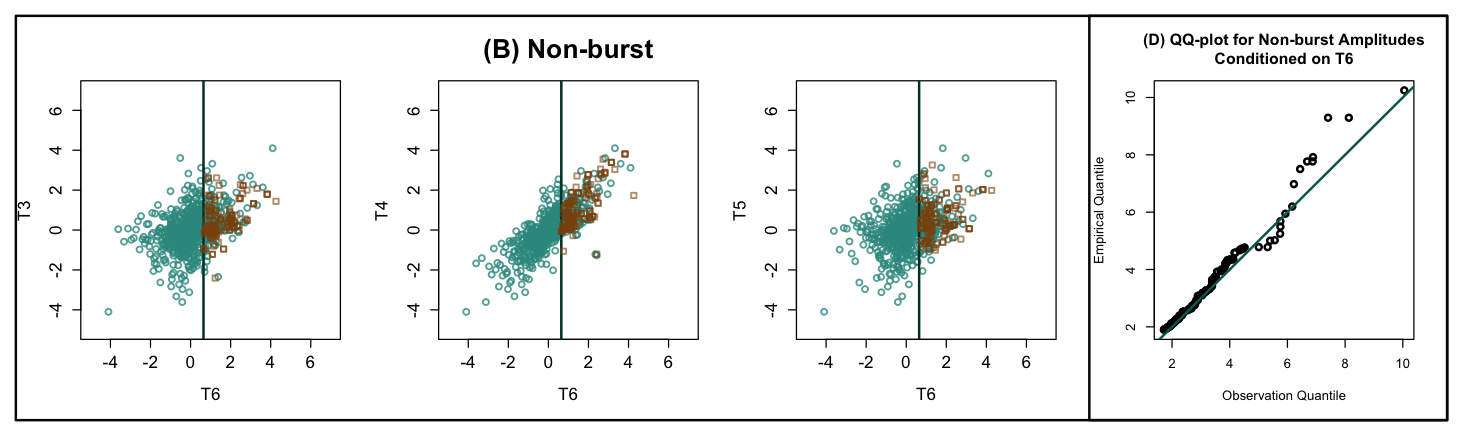}
    \caption{SpExCon estimates diagnostics for \textit{epileptic neonate 1}, conditioned on large values of channel T6.}
    \centering
    \includegraphics[width=1.01\textwidth]{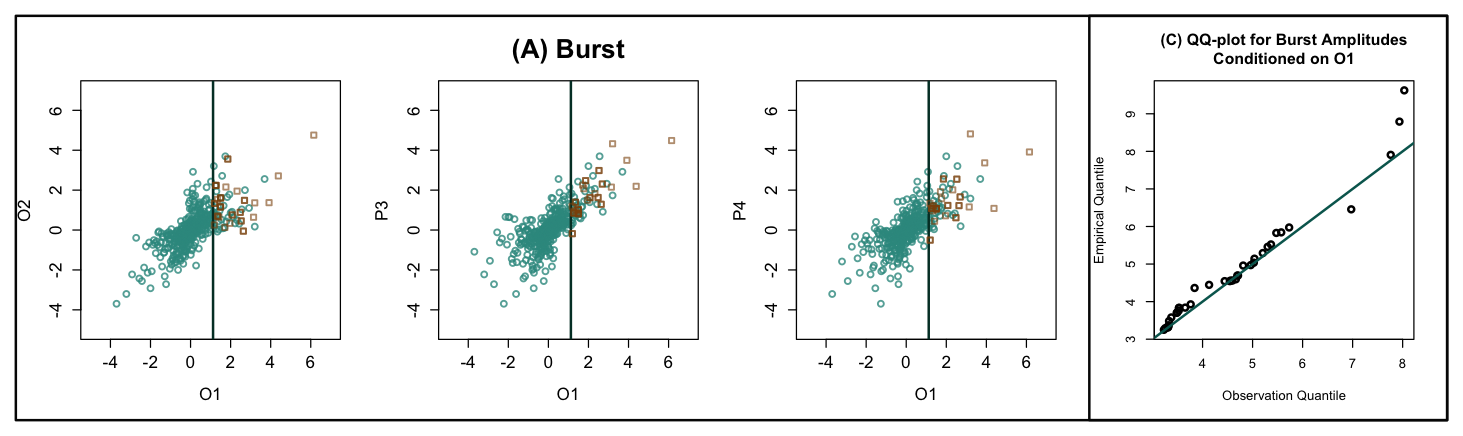}    \includegraphics[width=1.01\textwidth]{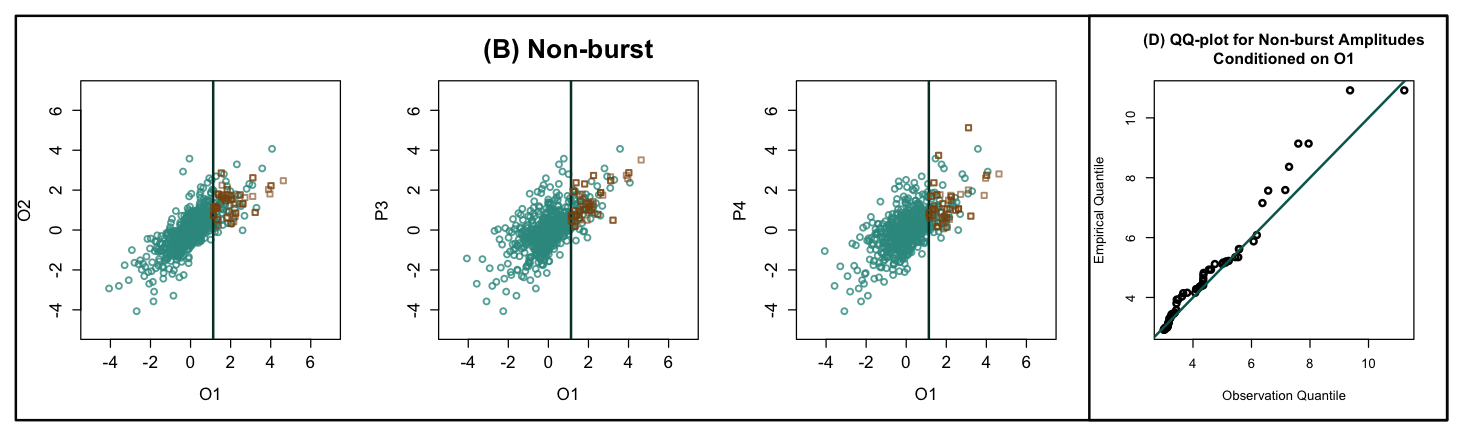}
    \caption{SpExCon estimates diagnostics for \textit{epileptic neonate 2}, conditioned on large values of channel O1.}
\end{figure}

\begin{figure}
    \centering
    \includegraphics[width=1.01\textwidth]{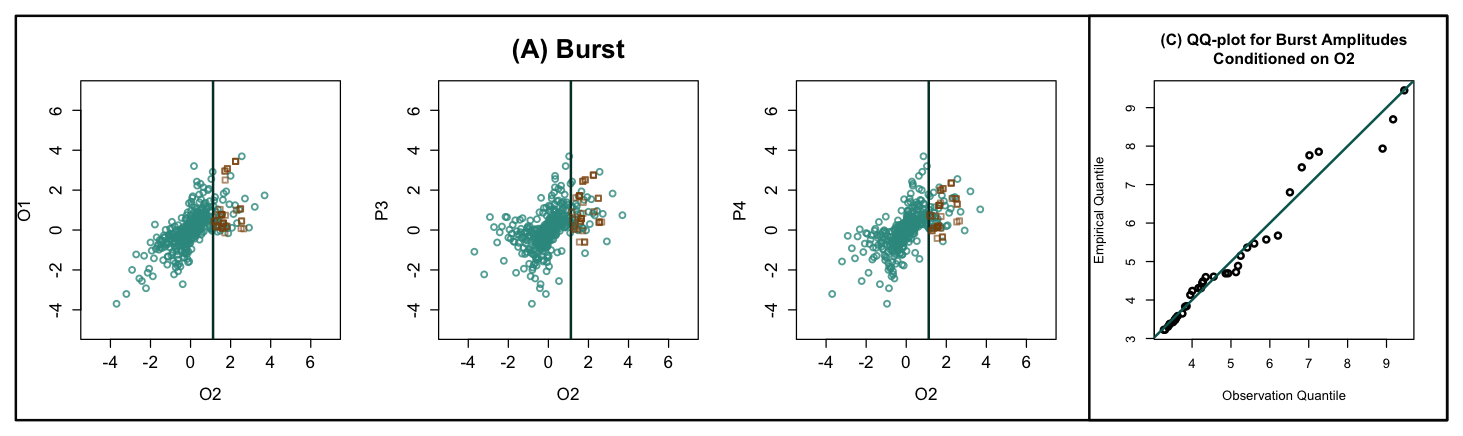}
    \includegraphics[width=1.01\textwidth]{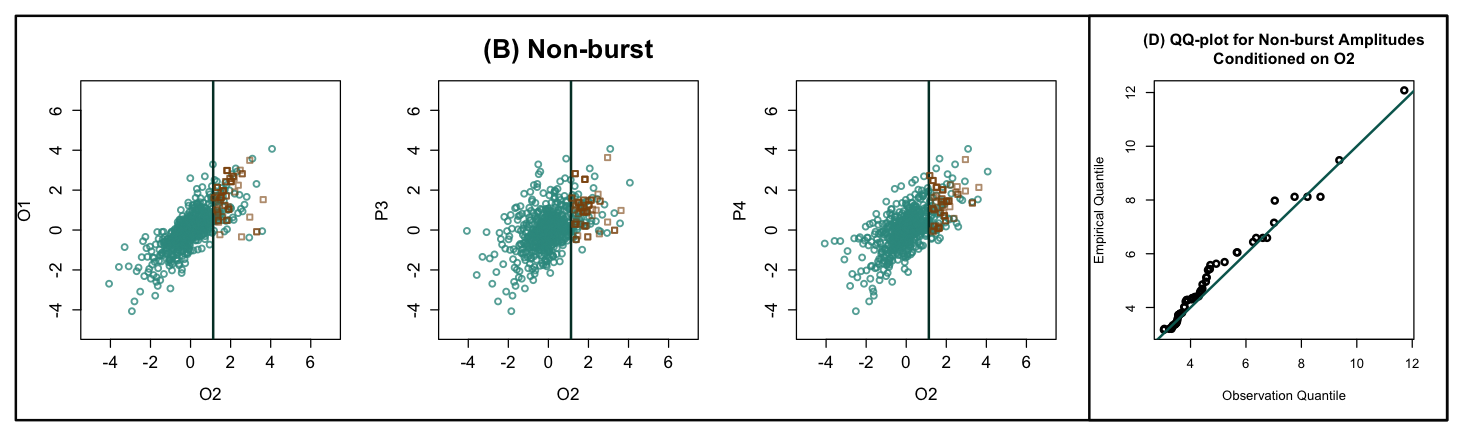}
    \caption{SpExCon estimates diagnostics for \textit{epileptic neonate 2}, conditioned on large values of channel O2.}
    \centering
    \includegraphics[width=1.01\textwidth]{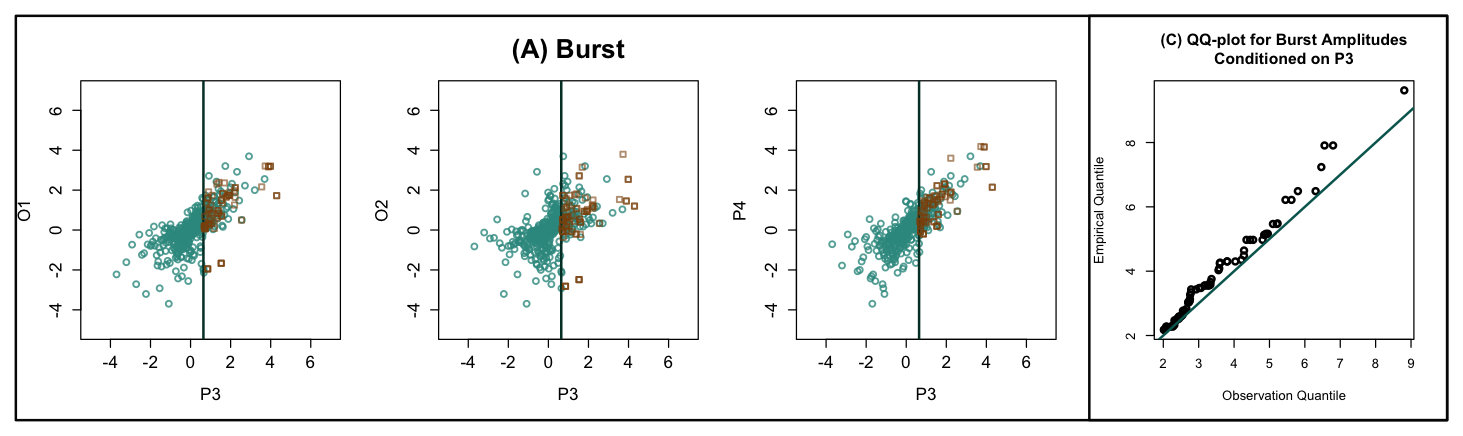}
    \includegraphics[width=1.01\textwidth]{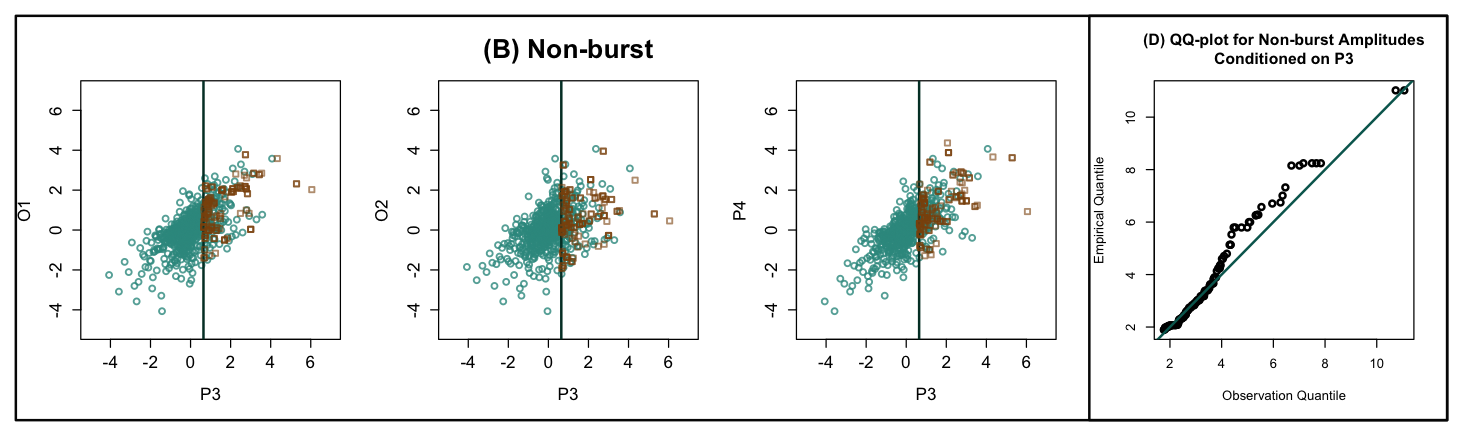}
    \caption{SpExCon estimates diagnostics for \textit{epileptic neonate 2}, conditioned on large values of channel P3.}
\end{figure}

\begin{figure}
    \centering
    \includegraphics[width=1.01\textwidth]{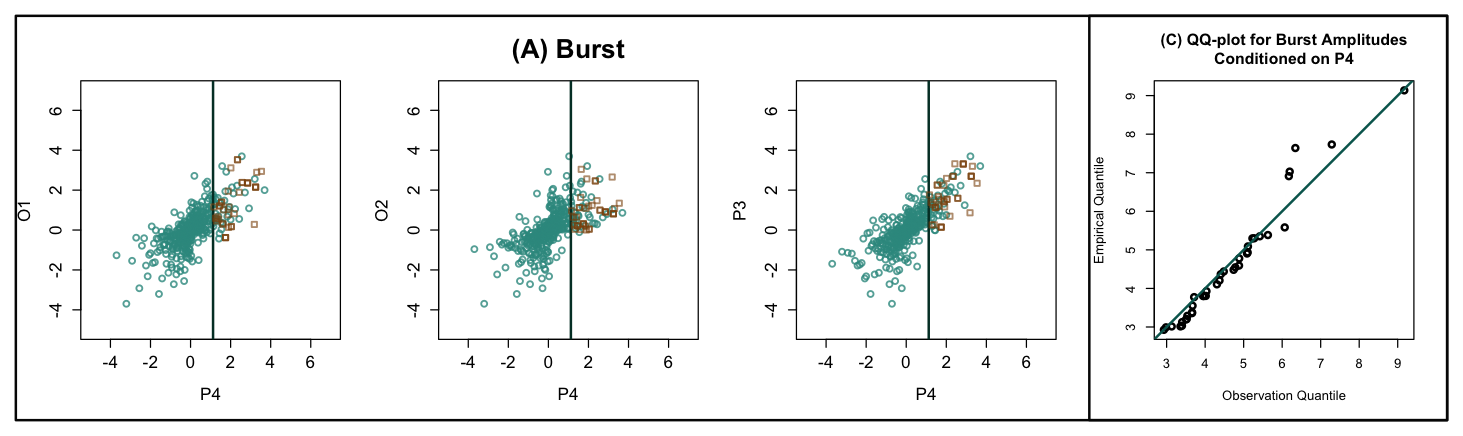}
    \includegraphics[width=1.01\textwidth]{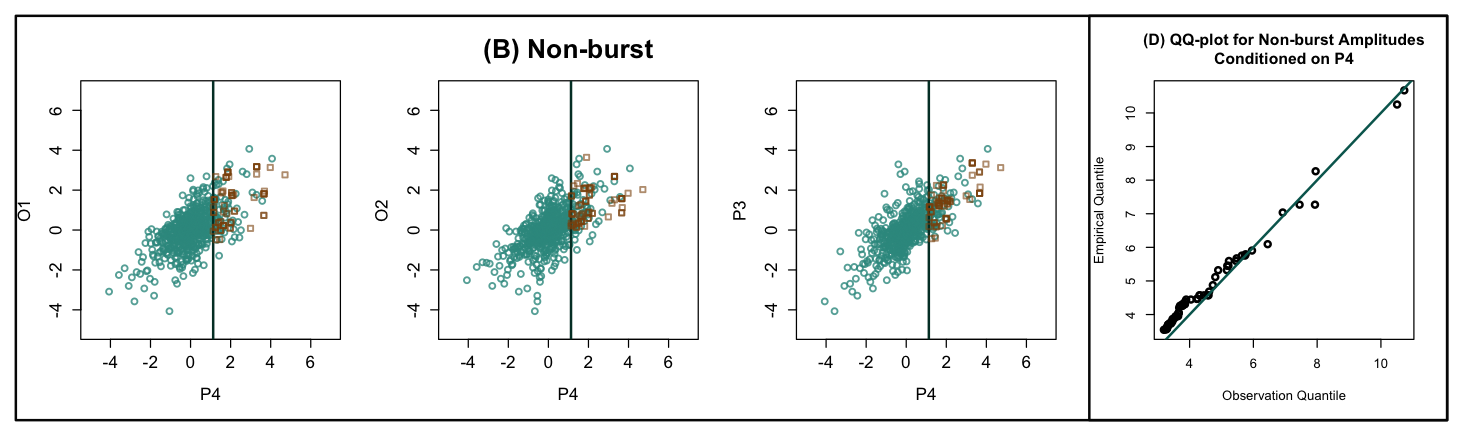}
    \caption{SpExCon estimates diagnostics for \textit{epileptic neonate 2}, conditioned on large values of channel P4.}
    \centering
    \includegraphics[width=1.01\textwidth]{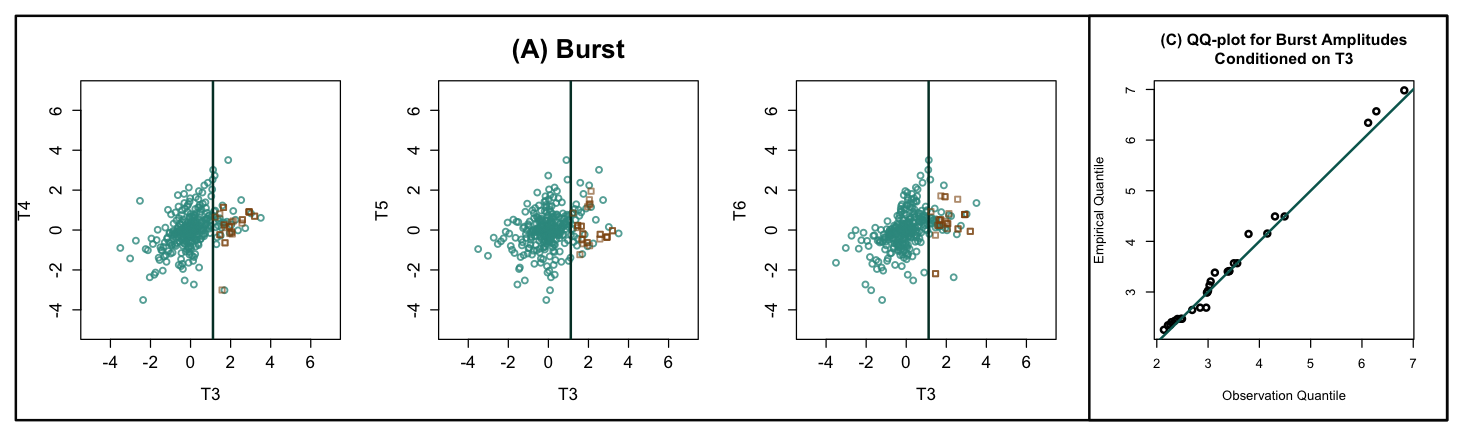}
    \includegraphics[width=1.01\textwidth]{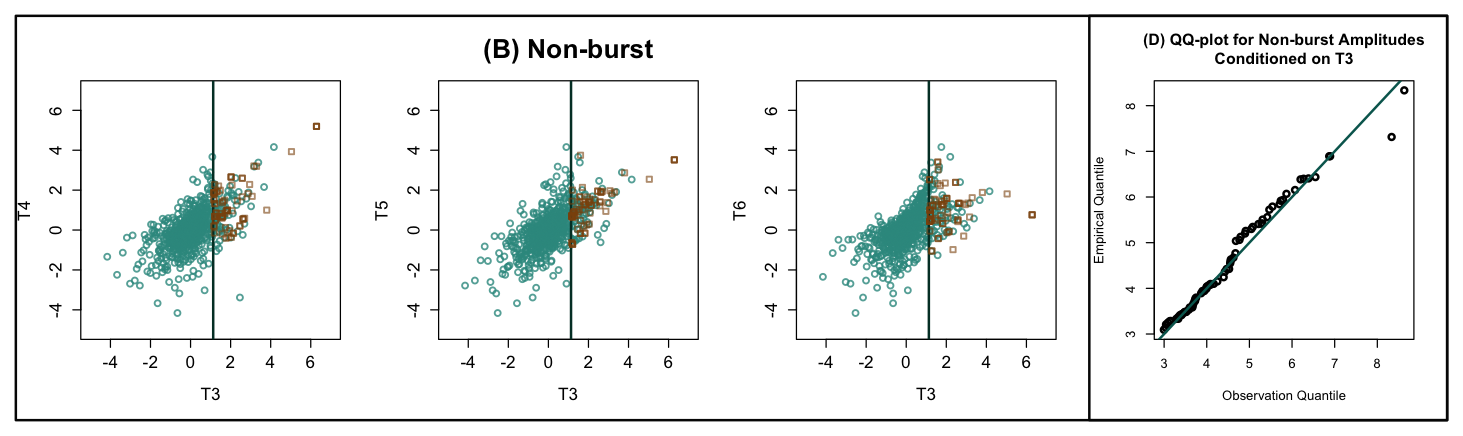}
    \caption{SpExCon estimates diagnostics for \textit{non-epileptic neonate 1}, conditioned on large values of channel T3.}
\end{figure}
    
\begin{figure}
    \centering
    \includegraphics[width=1.01\textwidth]{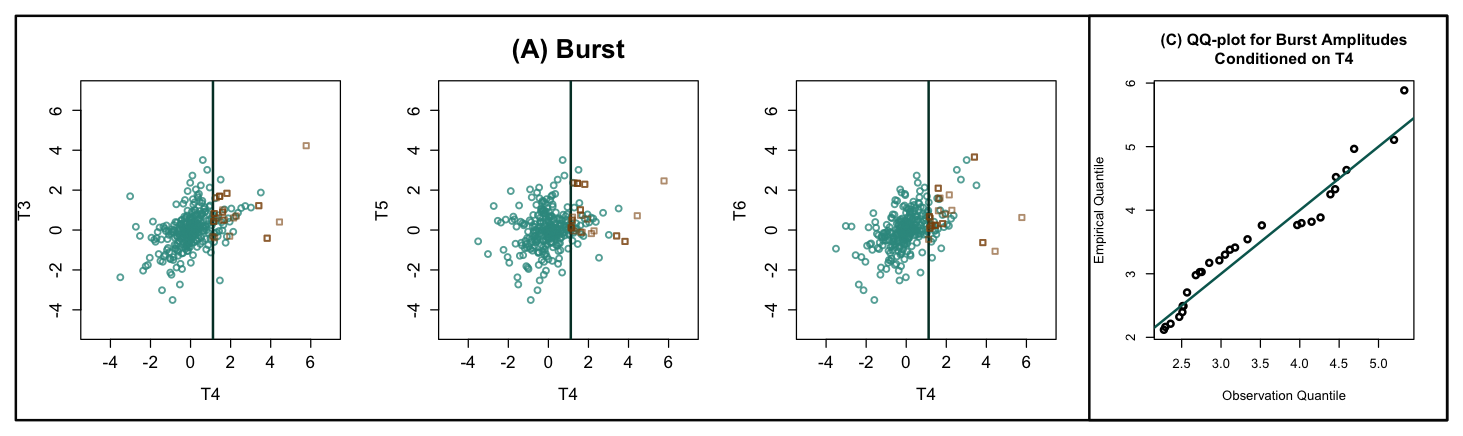}
    \includegraphics[width=1.01\textwidth]{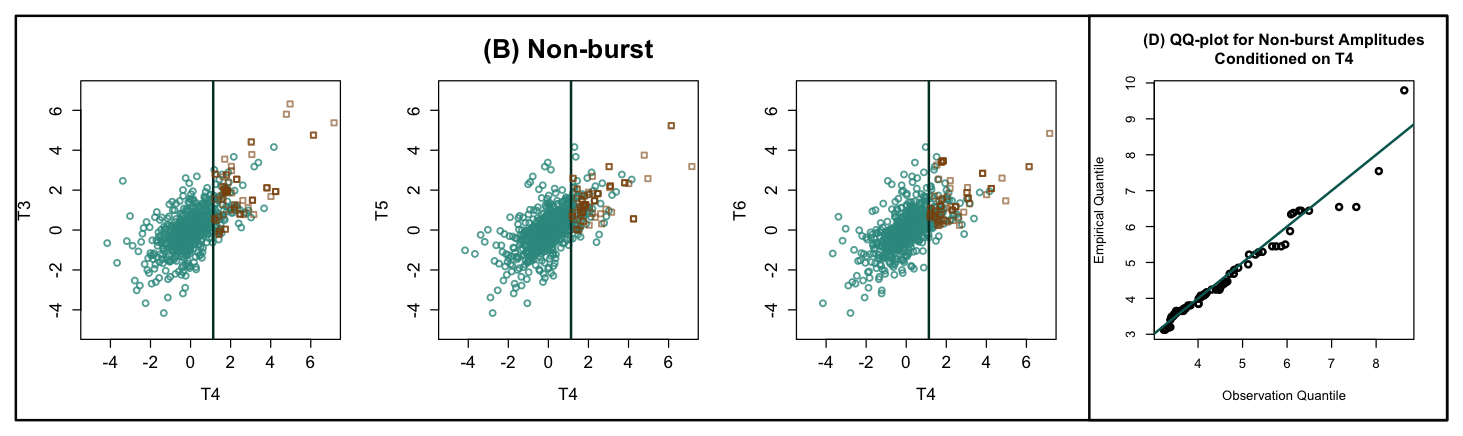}
    \caption{SpExCon estimates diagnostics for \textit{non-epileptic neonate 1}, conditioned on large values of channel T4.}
    \centering
    \includegraphics[width=1.01\textwidth]{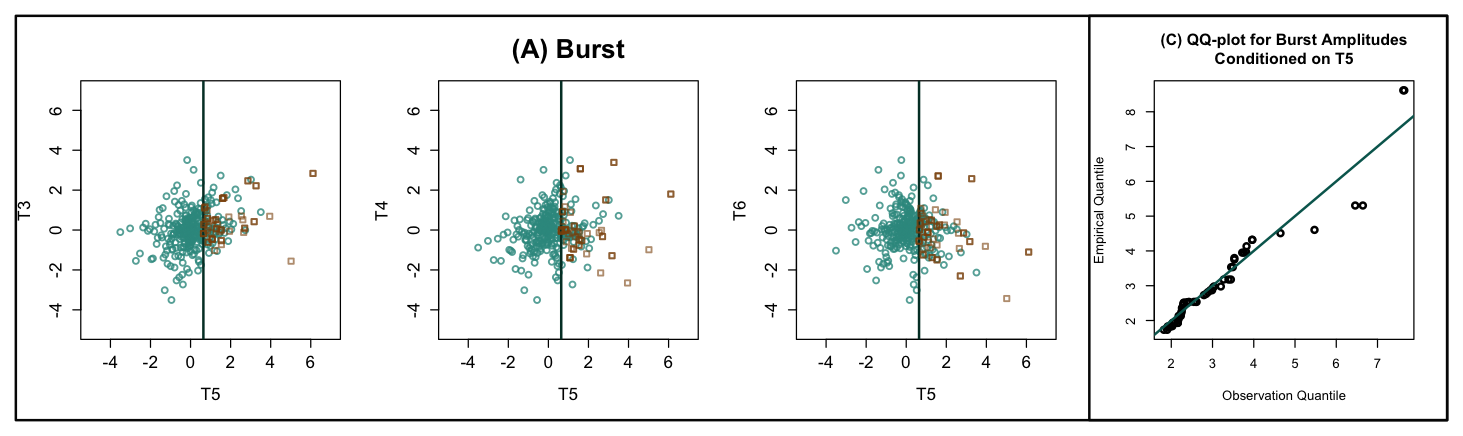}
    \includegraphics[width=1.01\textwidth]{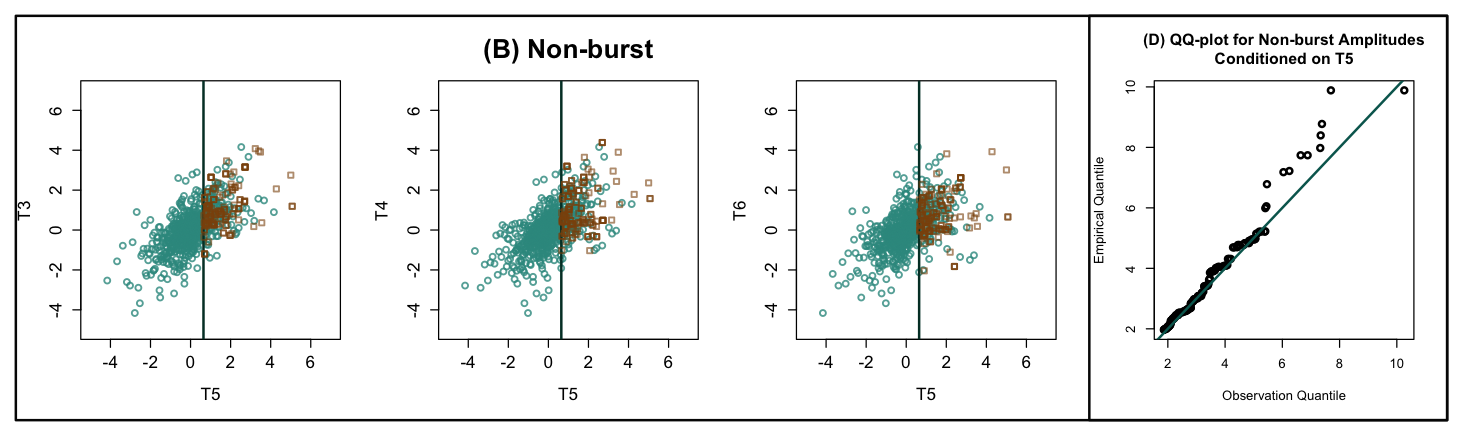}
    \caption{SpExCon estimates diagnostics for \textit{non-epileptic neonate 1}, conditioned on large values of channel T5.}
\end{figure}

\begin{figure}
    \centering
    \includegraphics[width=1.01\textwidth]{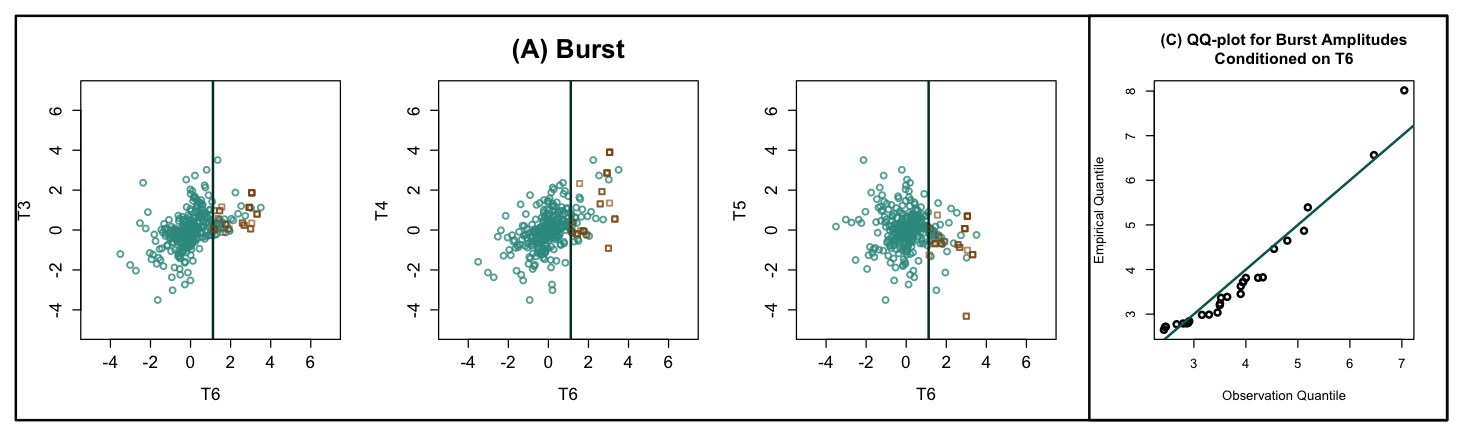}
    \includegraphics[width=1.01\textwidth]{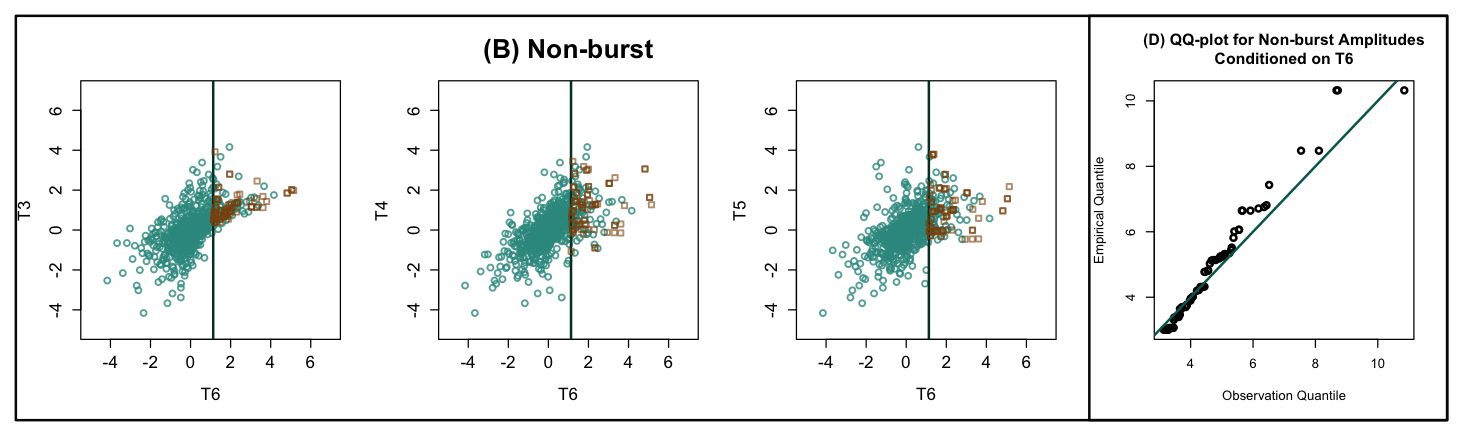}
    \caption{SpExCon estimates diagnostics for \textit{non-epileptic neonate 1}, conditioned on large values of channel T6.}
    \centering
    \includegraphics[width=1.01\textwidth]{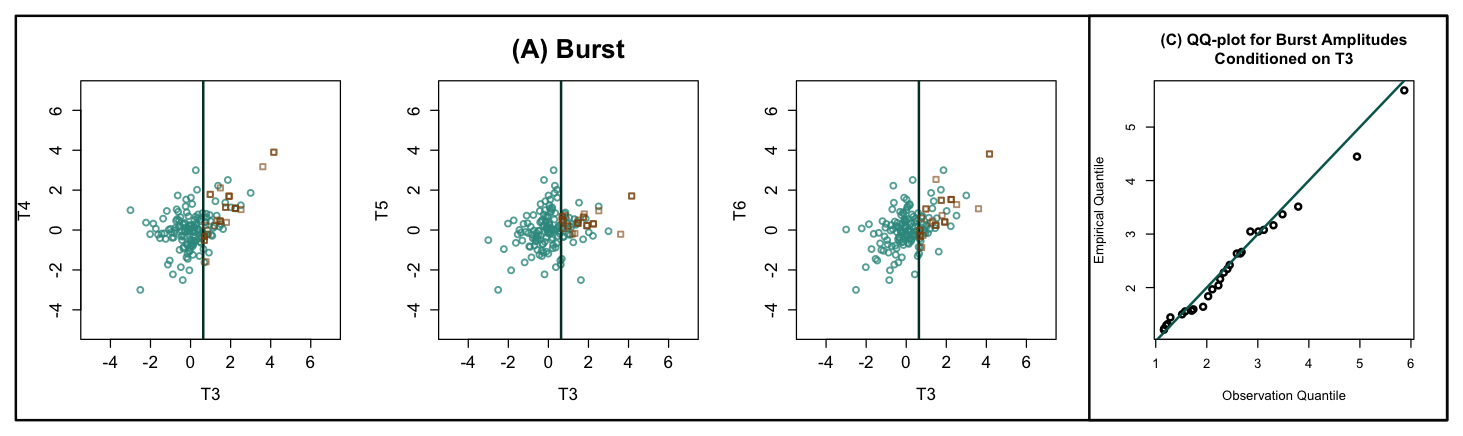}
    \includegraphics[width=1.01\textwidth]{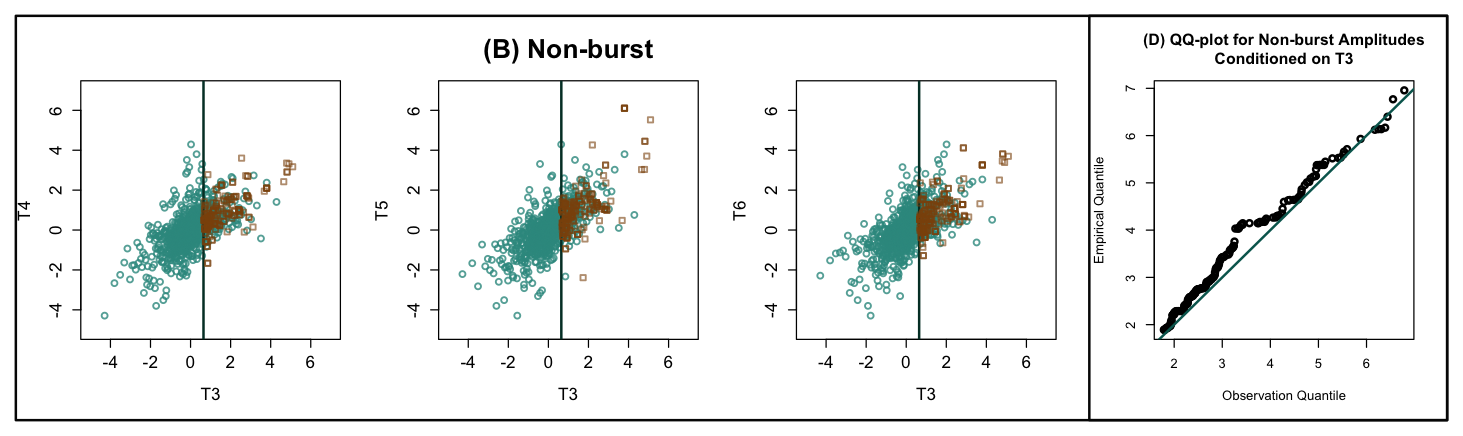}
    \caption{SpExCon estimates diagnostics for \textit{non-epileptic neonate 2}, conditioned on large values of channel T3.}
\end{figure}

\begin{figure}
    \centering
    \includegraphics[width=1.01\textwidth]{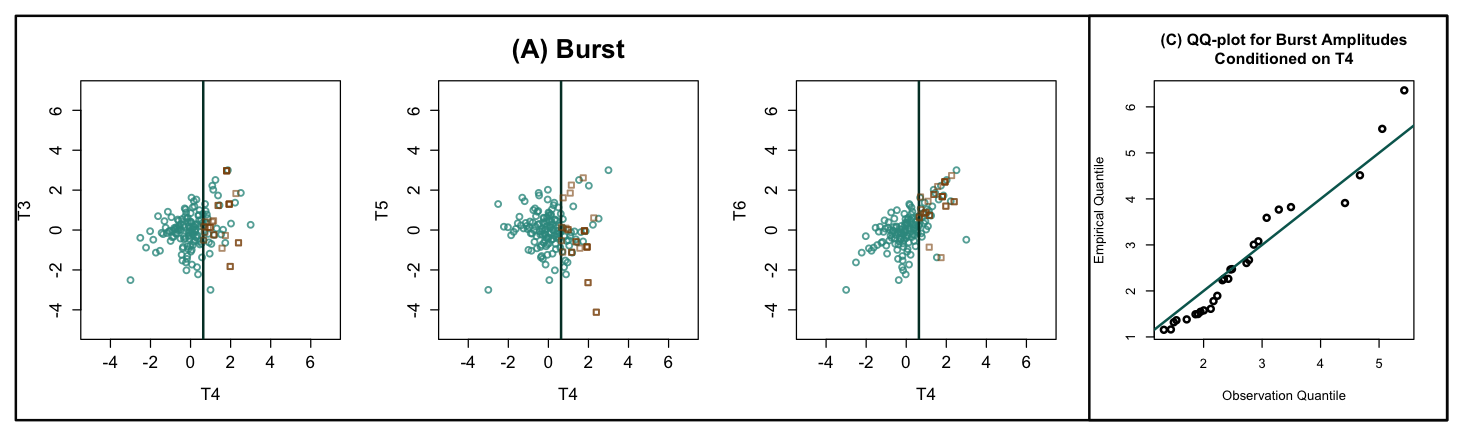}
    \includegraphics[width=1.01\textwidth]{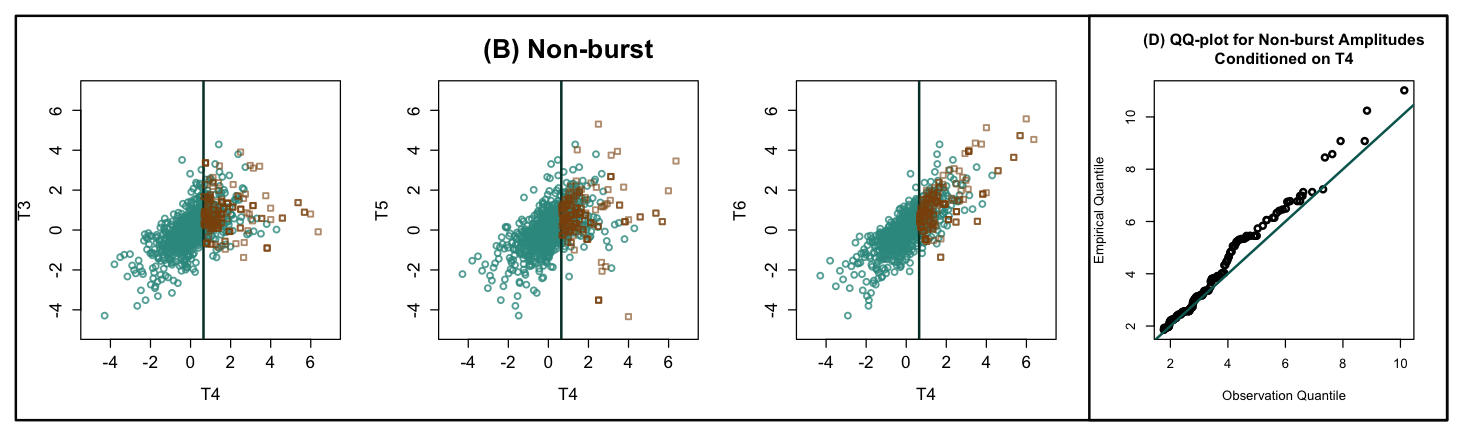}
    \caption{SpExCon estimates diagnostics for \textit{non-epileptic neonate 2}, conditioned on large values of channel T4.}
    \centering
    \includegraphics[width=1.01\textwidth]{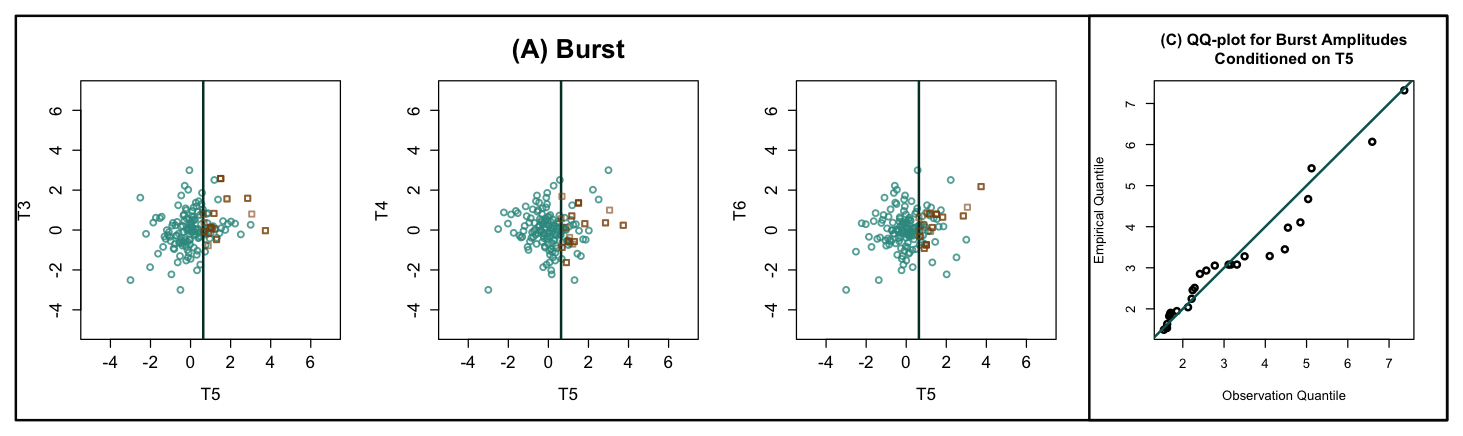}
    \includegraphics[width=1.01\textwidth]{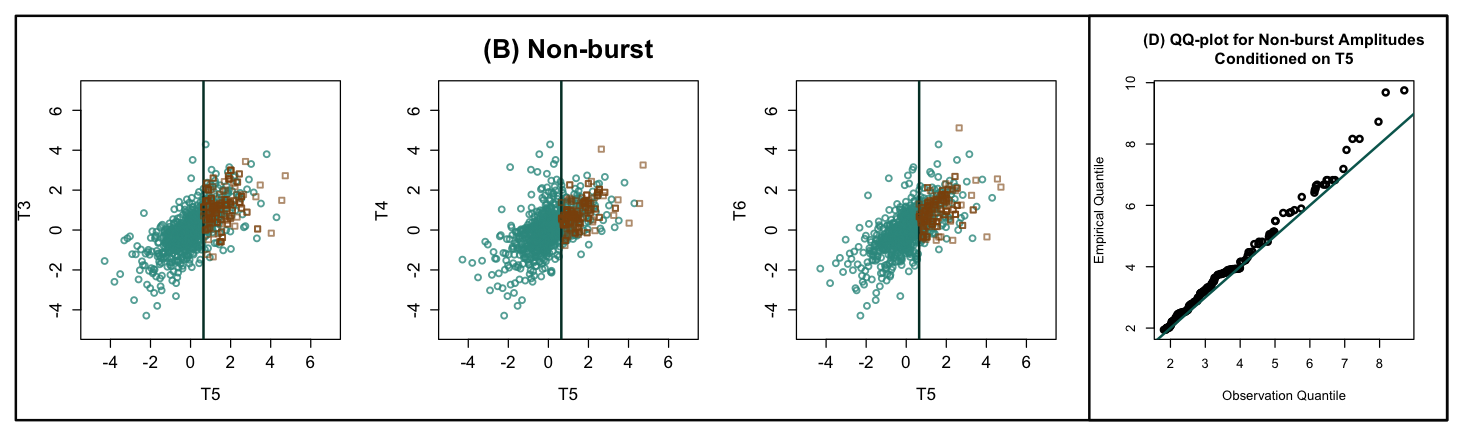}
    \caption{SpExCon estimates diagnostics for \textit{non-epileptic neonate 2}, conditioned on large values of channel T5.}
\end{figure}

\begin{figure}
    \centering
    \includegraphics[width=1.01\textwidth]{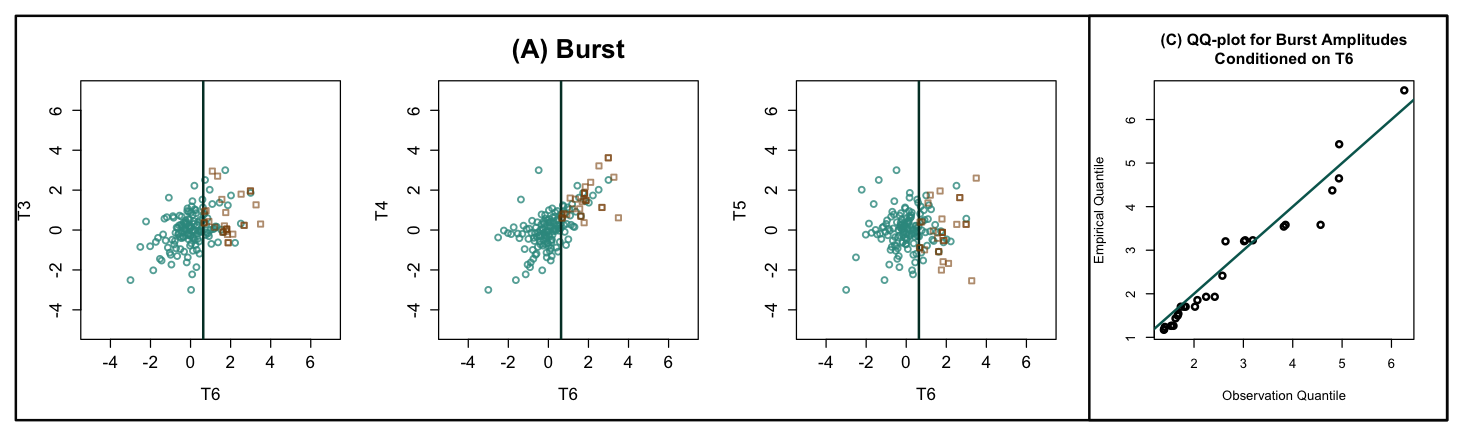}
    \includegraphics[width=1.01\textwidth]{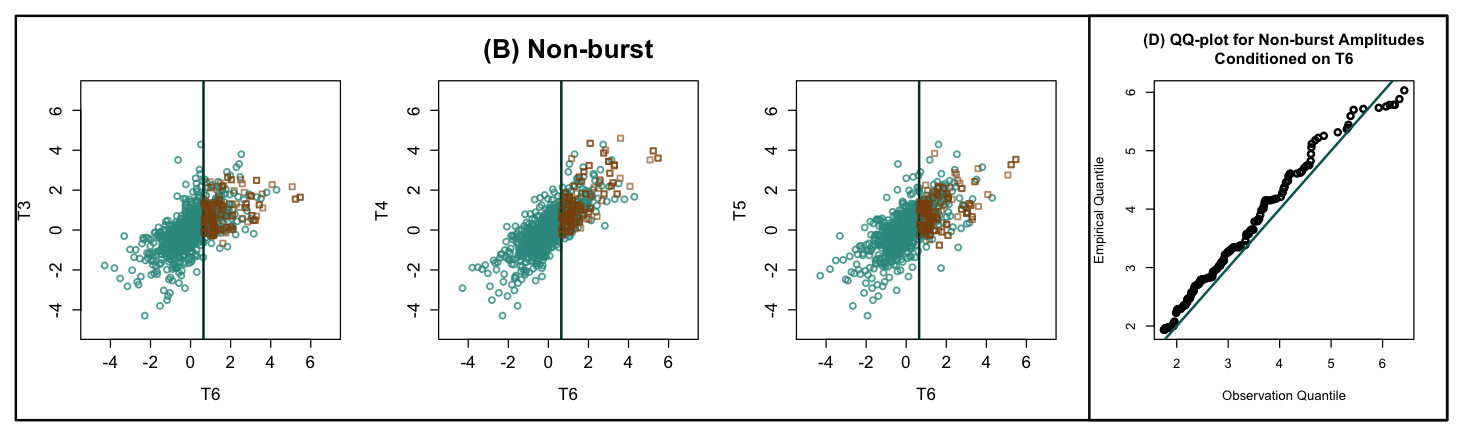}
    \caption{SpExCon estimates diagnostics for {\textit{non-epileptic neonate 2}}, conditioned on large values of channel T6.}
    \centering
    \includegraphics[width=1.01\textwidth]{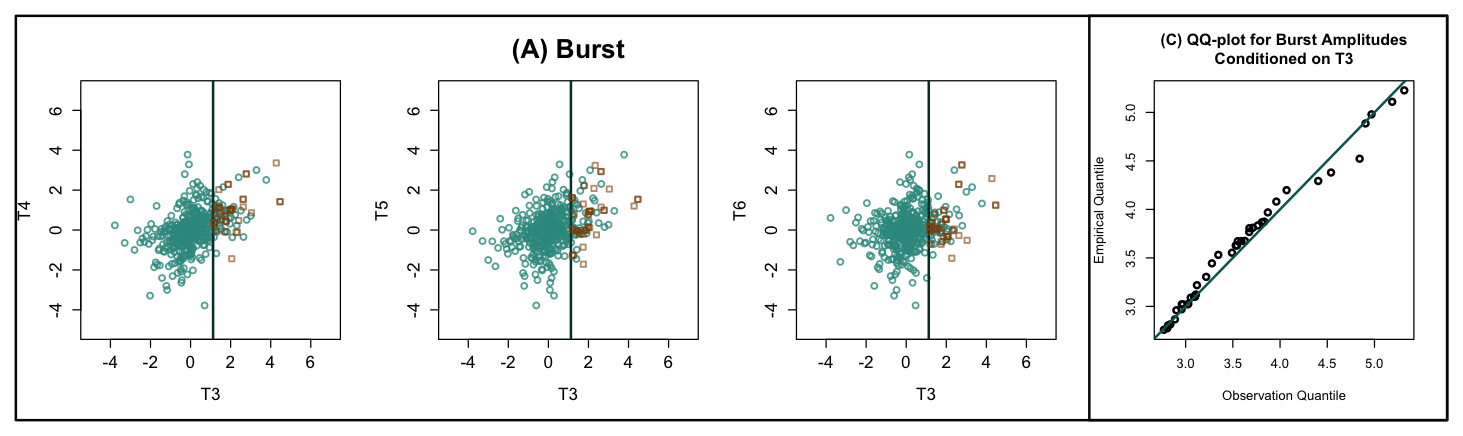}
    \includegraphics[width=1.01\textwidth]{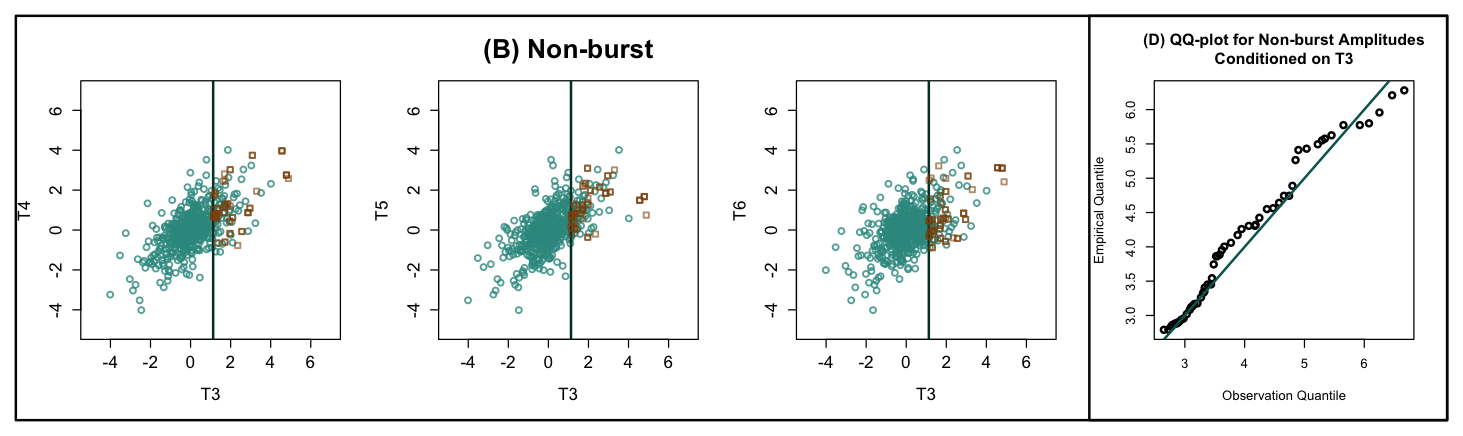}
    \caption{SpExCon estimates diagnostics for \textit{non-epileptic neonate 3}, conditioned on large values of channel T3.}
\end{figure}

\begin{figure}
    \centering
    \includegraphics[width=1.01\textwidth]{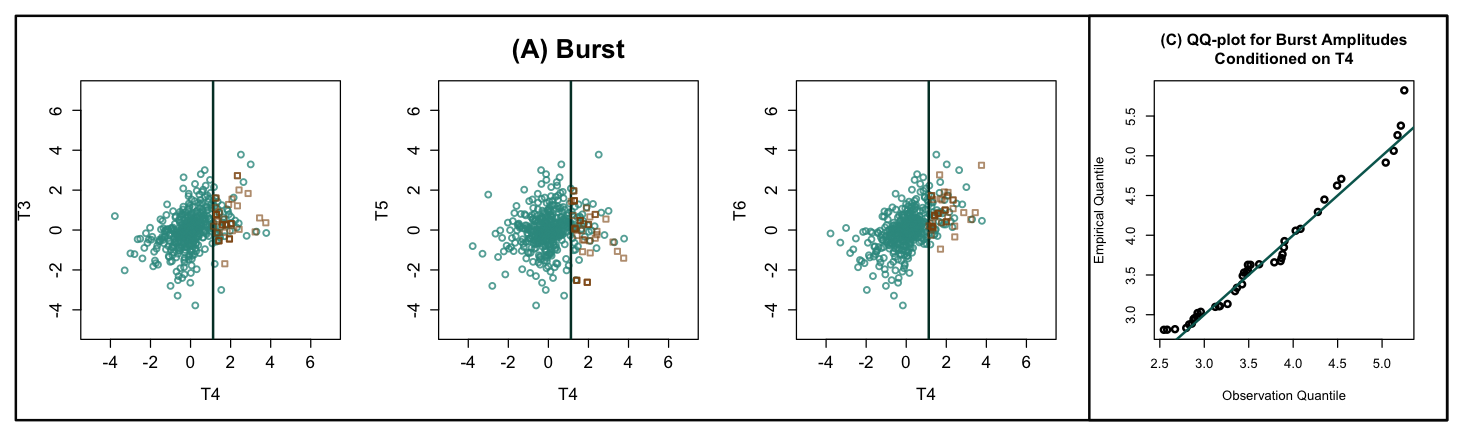}
    \includegraphics[width=1.01\textwidth]{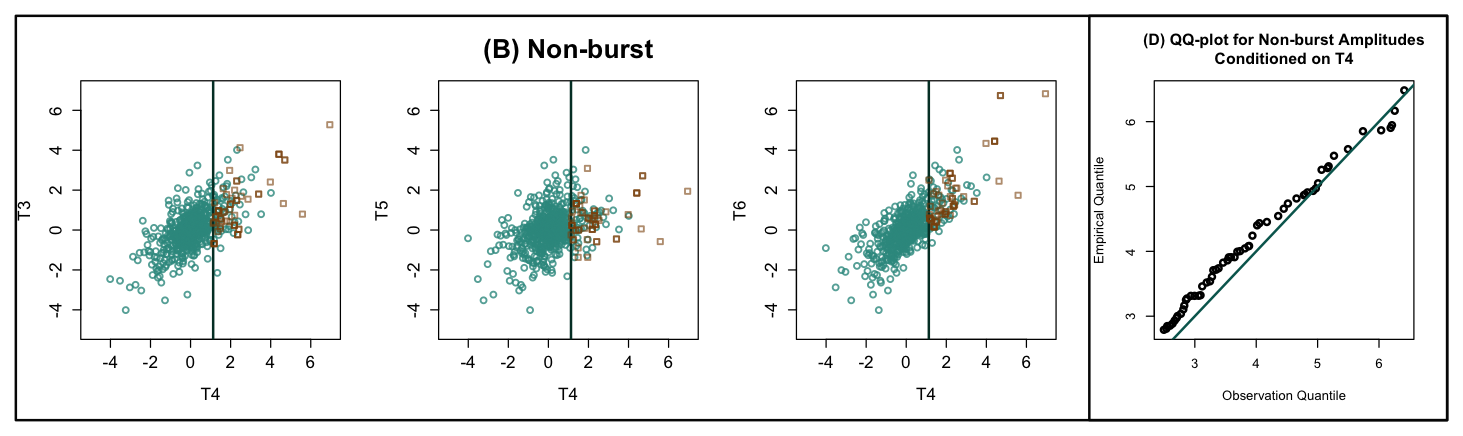}
    \caption{SpExCon estimates diagnostics for \textit{non-epileptic neonate 3}, conditioned on large values of channel T4.}
    \centering
    \includegraphics[width=1.01\textwidth]{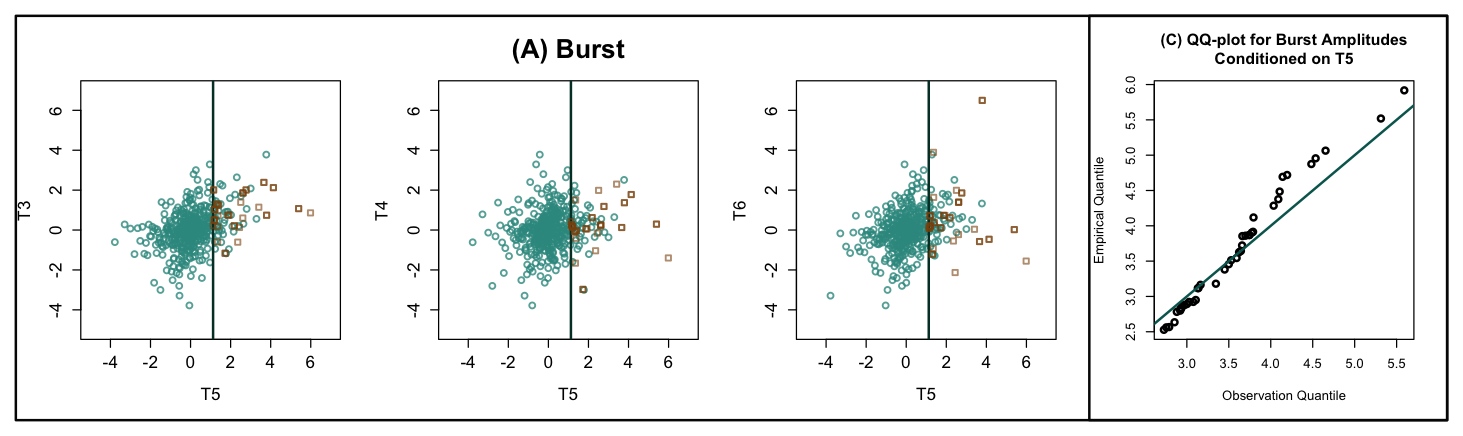}
    \includegraphics[width=1.01\textwidth]{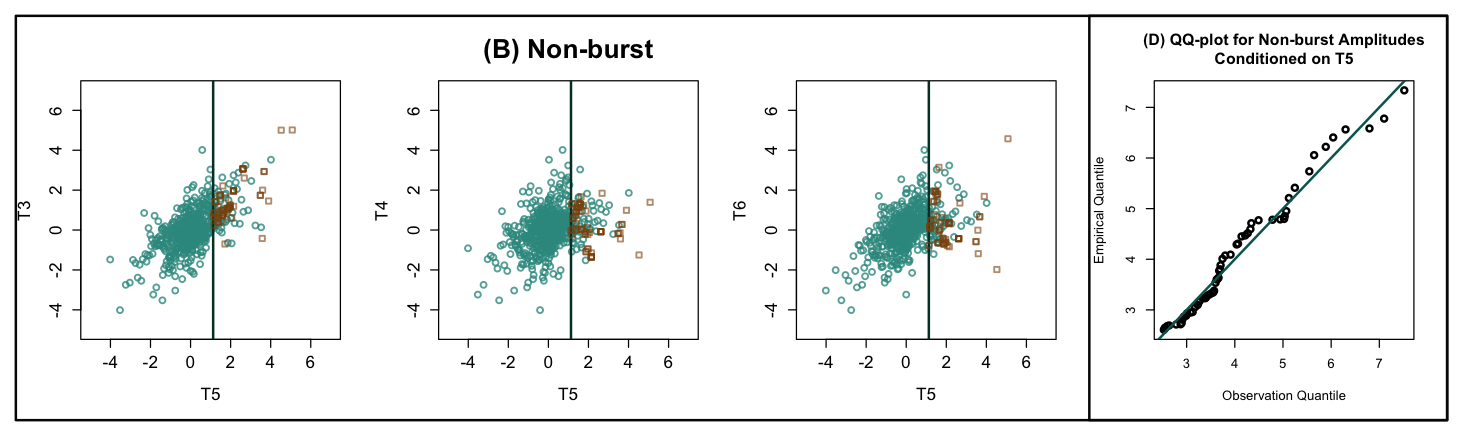}
    \caption{SpExCon estimates diagnostics for \textit{non-epileptic neonate 3}, conditioned on large values of channel T5.}
\end{figure}

\begin{figure}
    \centering
    \includegraphics[width=1.01\textwidth]{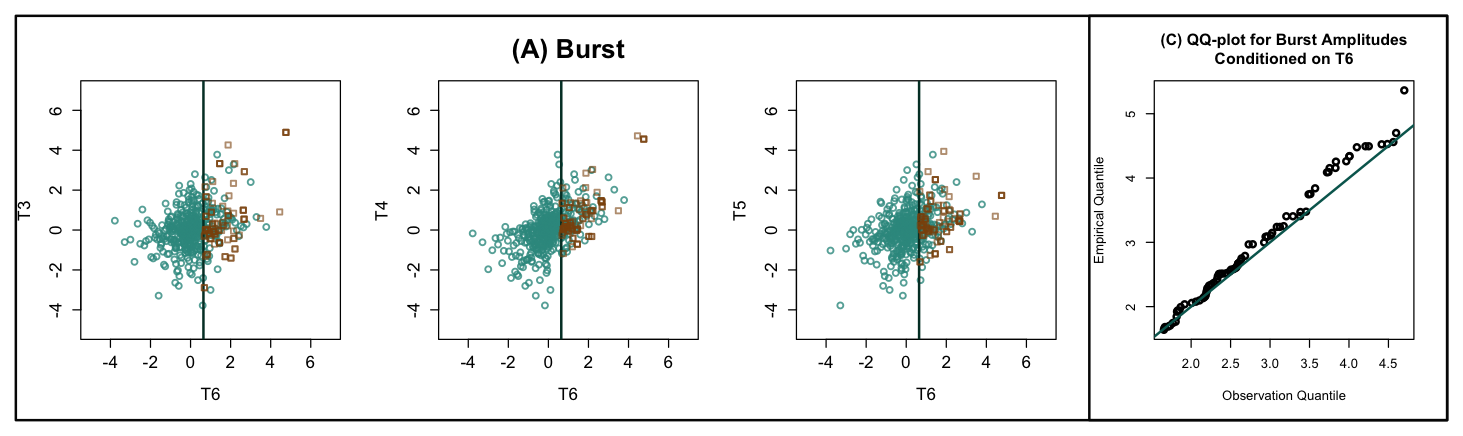}
    \includegraphics[width=1.01\textwidth]{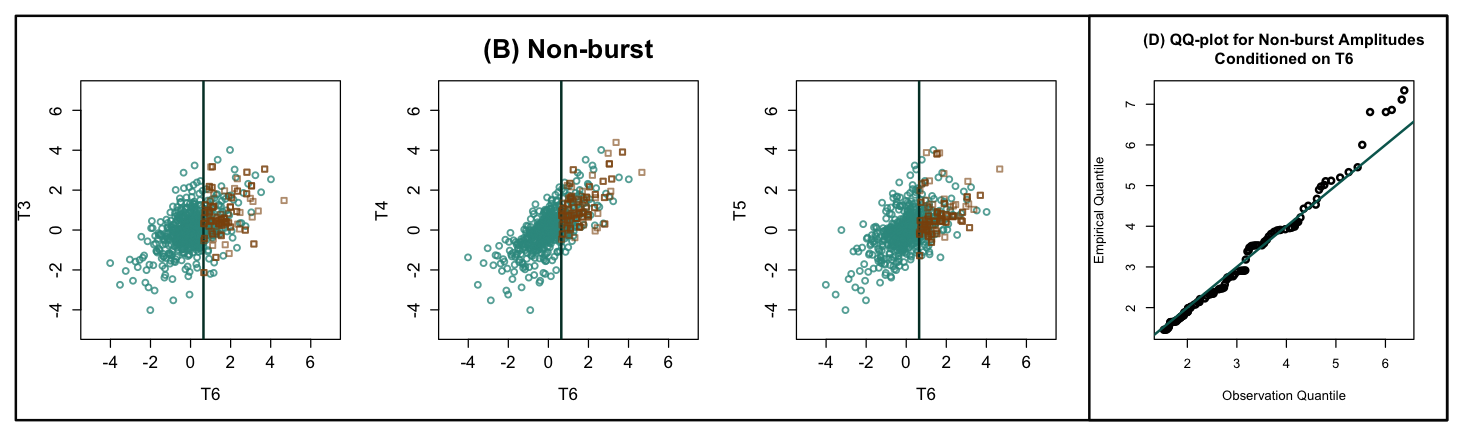}
    \caption{SpExCon estimates diagnostics for \textit{non-epileptic neonate 3}, conditioned on large values of channel T6.}
    
\end{figure}
       
\end{document}